\documentclass[svgnames, x11names, dvipsnames,xcolor=dvipsnames,12pt, draftclsnofoot, onecolumn]{IEEEtran}

\usepackage{enumitem}

\usepackage[dvipsnames]{xcolor}

\usepackage{amssymb}

\usepackage[most]{tcolorbox}
\tcbset{on line, 
        boxsep=4pt, left=0pt,right=0pt,top=0pt,bottom=0pt,
        colframe=white,colback=gray!10,  
        highlight math style={enhanced}
        }
        
\usepackage[noadjust]{cite}
\usepackage{filecontents}

\let\svtikzpicture\tikzpicture
\def\tikzpicture{\noindent\svtikzpicture}

\usepackage{amsbsy}
\usepackage{bm}
\usepackage{fixmath}
\usepackage{silence}
\usepackage{ctable}
\usepackage{pst-node,pst-plot}
\usetikzlibrary{shapes.geometric}
\usepackage{mathdots}
\usetikzlibrary{math}

\usepackage{upgreek}

\usepackage{autobreak}
\allowdisplaybreaks

 \usepackage{caption}
\usepackage{commath,amsmath,amssymb,amsfonts}
\usepackage{mathtools}
\usepackage{amsthm}
\usepackage{algorithmic}
\usepackage{pgfplots} 
\pgfplotsset{compat=1.15}
\usepackage{graphicx}
\usepackage{pgfgantt}
\usepackage{pdflscape}
\usepackage{pst-plot}
\usepackage{xfrac}
\usepackage{colortbl}
\usepackage{cancel}
\usepgfplotslibrary{fillbetween}
\usepackage{amssymb,bm}
\usepackage{float}
\usepackage{amsmath}
\usepackage[draft=false]{hyperref}
\usepackage{multirow}

\usepackage{mathrsfs}
\usepackage{bbm}

\newcommand{\floor}[1]{\left \lfloor #1 \right \rfloor}

\usepackage{cite} 

\usepackage{comment}

\usepackage{autobreak}
\allowdisplaybreaks

\def \D{\mathcal{D}}

\def \N{\mathcal{N}}

\def \P{\mathcal{P}}

\def \S{\mathcal{S}}

\def \fA{\mathbf{A}}
\def \fB{\mathbf{B}}

\def \fU{\mathbf{U}}
\def \fV{\mathbf{V}}

\def \fc{\mathbf{c}}

\def \ff{\mathbf{f}}

\def \fC{\mathbf{C}}
\def \fd{\mathbf{d}}

\def \fh{\mathbf{h}}
\def \fH{\mathbf{H}}
\def \fI{\mathbf{I}}

\def \fV{\mathbf{V}}
\def \fx{\mathbf{x}}
\def \fy{\mathbf{y}}
\def \fz{\mathbf{z}}

\def \fU{\mathbf{U}}

\def \fQ{\mathbf{Q}}
\def \fV{\mathbf{V}}
\def \fX{\mathbf{X}}
\def \fY{\mathbf{Y}}
\def \fZ{\mathbf{Z}}
\def \f0{\mathbf{0}}

\definecolor{blau_1a}{RGB}{93,133,195}
\definecolor{blau_2a}{RGB}{0,156,218}
\definecolor{gruen_3a}{RGB}{80,182,149}
\definecolor{gruen_4a}{RGB}{175,204,80}
\definecolor{gruen_5a}{RGB}{221,223,72}
\definecolor{orange_6a}{RGB}{255,224,92}
\definecolor{orange_7a}{RGB}{248,186,60}
\definecolor{rot_8a}{RGB}{238,122,52}
\definecolor{rot_9a}{RGB}{233,80,62}
\definecolor{lila_10a}{RGB}{201,48,142}
\definecolor{lila_11a}{RGB}{128,69,151}

\definecolor{blau_1b}{RGB}{0,90,169}
\definecolor{blau_2b}{RGB}{0,131,204}
\definecolor{gruen_3b}{RGB}{0,157,129}
\definecolor{gruen_4b}{RGB}{153,192,0}
\definecolor{gruen_5b}{RGB}{201,212,0}
\definecolor{orange_6b}{RGB}{253,202,0}
\definecolor{orange_7b}{RGB}{245,163,0}
\definecolor{rot_8b}{RGB}{236,101,0}
\definecolor{rot_9b}{RGB}{230,0,26}
\definecolor{lila_10b}{RGB}{166,0,132}
\definecolor{lila_11b}{RGB}{114,16,133}

\definecolor{mycolor1}{rgb}{0.0, 0.18, 0.39}
\definecolor{mycolor2}{RGB}{87,108,67}
\definecolor{mycolor3}{RGB}{8,133,161}
\definecolor{mycolor4}{RGB}{80,91,161}
\definecolor{mycolor5}{RGB}{98,122,157}
\definecolor{mycolor6}{RGB}{255,163,67}
\definecolor{mycolor7}{RGB}{152,205,225}
\definecolor{mycolor8}{RGB}{242,204,48}
\definecolor{mycolor9}{rgb}{0,.5,0}
\definecolor{mycolor10}{rgb}{.59,.44,.09}
\definecolor{mycolor11}{RGB}{231,199,31} % Yellow
\definecolor{mycolor12}{RGB}{8,133,161} % Cyan
\definecolor{mycolor13}{RGB}{157,188,64} % Yellow Green
\definecolor{mycolor14}{RGB}{194,150,130} % Light Skin
\definecolor{mycolor15}{RGB}{98,122,157} % Blue Sky
\definecolor{mycolor16}{RGB}{160,160,160} % Neutral
\definecolor{mycolor17}{RGB}{115,82,68} % Dark Skin
\definecolor{mycolor18}{RGB}{94,60,108} % Purple
\definecolor{mycolor19}{RGB}{115,82,68} % Dark Skin
\definecolor{mycolor20}{RGB}{255,183,30} % Dark Gold

\usepackage{accents}

\theoremstyle{remark} \newtheorem{theorem}{Theorem}
\theoremstyle{remark} 
\theoremstyle{remark} \newtheorem{definition}{Definition}
\theoremstyle{remark} 
\theoremstyle{remark} 

\newtheorem{innercustomgeneric}{\customgenericname}
\providecommand{\customgenericname}{}
\newcommand{\newcustomtheorem}[2]{%
  \newenvironment{#1}[1]
  {%
   \renewcommand\customgenericname{#2}%
   \renewcommand\theinnercustomgeneric{##1}%
   \innercustomgeneric
  }
  {\endinnercustomgeneric}
}
\newcustomtheorem{customcorollary}{Corollary}
\newcustomtheorem{customlemma}{Lemma}
\newcustomtheorem{customproposition}{Proposition}

\usepackage{amssymb}

\usepackage{autobreak}
\allowdisplaybreaks

\usepackage{tikz}
\usetikzlibrary{arrows.meta}

\usepackage{blochsphere}

\usepackage{multicol, blindtext}

\usepackage[]{pgfplots}
\usepackage{pgfplotstable}
\usepgfplotslibrary{dateplot}
\pgfplotsset{/pgf/number format/use comma,compat=newest}

\usepackage{amsbsy}
\usepackage{bm}
\usepackage{fixmath}
\usepackage{silence}
\usepackage{ctable}
\usepackage{pst-node,pst-plot}
\usetikzlibrary{shapes.geometric}
\usepackage{mathdots}
\usetikzlibrary{math}

\usepackage{feyn}
\usepackage{mathtools,multirow}

\def\fc{{\bf c}}
\def\fA{{\bf A}}
\def\fU{{\bf U}}
\def\fV{{\bf V}}

\def\fB{{\bf B}}
\def\fx{{\bf x}}

\def\fSigma{\boldsymbol{\mathbf{\Sigma}}}

\usetikzlibrary{arrows, arrows.meta}
\usepackage{pgfplots}
\pgfplotsset{compat=1.15}
\usepgfplotslibrary{polar}

\usepackage[bottom]{footmisc}

\usepackage{cleveref}
\crefname{equation}{Eq}{} % capitalize "E", no period
\crefrangelabelformat{equation}{(#3#1#4--#5#2#6)}

\usepackage{tcolorbox}

\usepackage{lipsum}   % To generate test text 
\usepackage{framed}
\usepackage{tikz}
\usetikzlibrary{decorations.pathmorphing,calc}
\pgfmathsetseed{1} % To have predictable results
\pgfdeclarelayer{background}
\pgfsetlayers{background,main}

\def \bG {\mathcal{G}_{\fh}}

\usepackage{nccmath}

\usepackage[svgnames]{xcolor}
\usepackage{tikz}
\usepackage{mathtools}

\def\nudge{.5}

\tikzset{axis/.style={ultra thick, Red!75!black, -latex, shorten <=-\nudge cm, shorten >=-2*\nudge cm}}
\tikzset{line/.style={thick,Green}}
\allowdisplaybreaks

\def\fc{{\bf c}}
\def\fA{{\bf A}}
\def\fU{{\bf U}}
\def\fV{{\bf V}}

\def\fB{{\bf B}}
\def\fx{{\bf x}}

\def\flambda{\boldsymbol{\mathbf{\Uplambda}}}
\def\fSigma{\boldsymbol{\mathbf{\Sigma}}}

\makeatletter
\patchcmd{\@maketitle}
{\addvspace{0.5\baselineskip}\egroup}
{\addvspace{-.5\baselineskip}\egroup}
{}
{}
\usepackage[letterpaper,bindingoffset=0.2in,
left=.7in,right=.7in,top=.9in,bottom=1.1in,
footskip=.2in]{geometry}

\usepackage{xpatch}
\xapptocmd{\appendix}{%
}{}{\PatchFailed}

\begin{document}
	
\title{Identification for Colored Gaussian Channels}
\author{\vspace{5mm} \fontsize{13}{13} \selectfont Mohammad Javad Salariseddigh\IEEEauthorrefmark{1} and Christian Deppe\IEEEauthorrefmark{2}
	\vspace{7mm}
	\\
	\fontsize{11}{11} \selectfont \IEEEauthorrefmark{1}Technical University of Darmstadt
	\\
	\vspace{1mm}
	\fontsize{11}{11} \selectfont\IEEEauthorrefmark{2}Technical University of Braunschweig
	\\
	% Email: \{m.j.salariseddigh@rcs.tu-darmstadt.de,christian.deppe@tu-braunschweig.de\}
}

	\maketitle
	
	\IEEEpeerreviewmaketitle
	
	% ---
	\begin{abstract}
		We study the identification capacity of discrete-time Gaussian channels impaired by correlated noise and inter-symbol interference (ISI). Our analysis is formulated for deterministic encoding functions subject to a peak power constraint and colored noise whose covariance matrix features a polynomially bounded singular value spectrum, i.e., $\sim [n^{-\mu} , n^{\mu/2}]$ where $n$ is the codeword length and $\mu \in [0,1/2)$ is the spectrum rate.  A central result establishes that, even when the ISI memory length grows sub-linearly with $n,$ i.e., $\sim n^{\kappa}$ where $\kappa \in [0,1/2)$ and $\kappa + \mu \in [0,1/2),$ the codebook size continues to exhibit super-exponential growth in $n$, i.e., $\sim 2^{(n \log n)R},$ with $R$ representing the associated coding rate. Moreover, by employing the well-known Mahalanobis-distance decoder induced by colored Gaussian noise statistics, we characterize bounds on the identification capacity, with the resulting bounds parameterized by $\kappa$ and $\mu.$
	\end{abstract}
	% ---
	
	\IEEEpeerreviewmaketitle
	
	\vspace{0mm}
	% ---
	
	\section{Introduction}
	
	In the identification setting \cite{J85,AD89,AH08G}, encoding and decoding schemes are designed such that the receiver can decide, with vanishing error probabilities, whether a given message of interest was transmitted. In contrast to Shannon’s classical communication model \cite{S48}, which requires reliable reconstruction of the transmitted message from the entire message set, the identification framework restricts attention to a single pre-specified message, thereby reducing decoding to a binary hypothesis test on its presence.
	%Identification problem with and without randomness at the encoder for discrete alphabet channels under average power constraint is studied in \cite{AD89} and \cite{Salariseddigh_IT}. While discrete channels resemble the same codebook size scaling, i.e., of the order $\sim 2^{nR}.$
	A well-known phenomenon for deterministic identification (DI) \cite{AN99,Salariseddigh_PhD_Diss} across continuous-alphabet channels, including the Gaussian channel with fading \cite{Salariseddigh_IT,Vorobyev25,Yuan22}, Poisson channels with and without inter-symbol interference (ISI) \cite{Salariseddigh-TMBMC,Salariseddigh_OJCOMS_23}, affine Poisson channels \cite{Salariseddigh25_ITW}, and binomial channels \cite{Salariseddigh_Binomial_ISIT}, is the emergence of a \emph{super-exponential} codebook size scaling, i.e., of the order $\sim 2^{(n\log n)R}.$ The identification has received considerable attention in post-Shannon and semantic communication frameworks \cite{Salariseddigh23_BSC_Future_Internet}. Identification code constructions are discussed in \cite{Verdu93,Lengerke25,Zinoghli24,Derebeyouglu2020}. Generalized models of identification problem and their connection to the Shannon problem are discussed in \cite{AH08G,AADT20,Rosenberger2024,Rosenberger25}.
	
	The inter-symbol interference (ISI) Gaussian channel with colored noise constitutes a canonical model for modern wireless communication systems \cite{Proakis01,G05}. In this setting, temporal correlation in the noise, induced by filtering, co-channel interference, and hardware impairments, interacts with channel memory due to ISI, yielding a nontrivial impact on both capacity characterization and receiver design. From an information-theoretic standpoint, this interplay requires coding and decoding strategies that explicitly accommodate memory in both the channel and the noise, thereby guiding the design of robust communication schemes. The Shannon capacity of colored Gaussian channels with ISI is classically achieved via water-filling over the channel spectrum, as established by Gallager in \cite{RG68}. Subsequent work characterized the capacity of discrete-time Gaussian ISI channels under per-symbol average power constraints \cite{Hirt88,Hirt02}. Extensions to multiuser scenarios include the capacity region of the Gaussian broadcast channel with ISI and colored noise under input power constraints \cite{Goldsmith02}, as well as the capacity region of the two-user Gaussian multiple-access channel with ISI \cite{Cheng93}. More recently, attention has turned to models with stochastic and time-varying channel coefficients, further enriching the theoretical landscape \cite{Moshksar24}.

	In this paper, we study the identification problem over the Gaussian channels with correlated noise and ISI employing a deterministic encoder in the presence of peak power constraint. We note that as its special case, the color noise Gaussian channel can model white Gaussian channel, by choosing $\mu=0$ \cite{Salari26}. While identification capacity has been studied for ISI-free channels \cite{Salariseddigh_IT} and under white-noise assumptions \cite{Salari26}, to the best of the author's knowledge it has not yet been characterized for the general Gaussian channel with intersymbol interference and colored noise.

	\subsection{Notations}
	Blackboard letters $\mathbbmss{X,Y,Z},\ldots$ denote alphabet sets. Set difference is written as $\mathbbmss{X} \setminus \mathbbmss{Y}$. Lowercase letters $x, y, z, \ldots$ denote realizations of random variables (RVs), while uppercase letters $X, Y, Z, \ldots$ denote RVs. Lowercase bold symbols $\fx,\fy,\fz,\ldots$ represent row vectors. All logarithms are taken to base $2$. The set of integers from $1$ to $M$ is denoted by $[\![M]\!]$. The sets of non-negative real numbers and real numbers are denoted by $\mathbb{R}_{+}$ and $\mathbb{R}$, respectively. The Gamma function $\Gamma(x)$ for a non-negative integer $x$ is given by $\Gamma (x) = (x-1) ! \triangleq (x-1) \times \dots \times 1$. We adopt the standard \emph{Bachmann--Landau asymptotic notation}, including little-$o$ notation $o(\cdot)$, big-$O$ notation $\mathcal{O}(\cdot)$, and asymptotic equivalence $\sim$. The expectation of an RV $X$ is denoted by $\mathbb{E}[X]$. The variance of a real-valued RV $X$ with finite first moment is defined as $\text{Var}[X] = \mathbb{E}[(X-\mathbb{E}[X])^2]$, and the covariance between two real-valued RVs $X$ and $Y$ with finite second moments is defined as $\text{Cov}[X,Y] = \mathbb{E}[(X-\mathbb{E}[X])(Y-\mathbb{E}[Y])]$. For a random vector $\fX$, the expectation is defined as $\mathbb{E}[\fX]$. The covariance matrix of $\fX$ is defined as $\mathrm{Cov}[\fX]=\mathbb{E}[(\fX-\mathbb{E}[\fX])(\fX-\mathbb{E}[\fX])^T] = \mathbb{E}[\fX\fX^T] - \mathbb{E}[\fX]\mathbb{E}[\fX]^T$. The quantities $\norm{\mathbf{x}}$ and $\norm{\mathbf{x}}_{\infty}$ denote the $\ell_2$-norm and $\ell_{\infty}$-norm, respectively. An $n$-dimensional hypersphere with radius $r$ and center $\mathbf{x}_0$ is defined as $\S_{\fx_0}(n,r) = \{ \fx \in \mathbb{R}^n : \norm{\fx-\fx_0} \leq r \}$. The set $\mathbbmss{Q}_{\f0} = \{\fx \in \mathbb{R}^n : | x_t | \leq U, \forall \, t \in [\![n]\!] \}$ denotes an $n$-dimensional hypercube with edge length $U$ centered at the origin $\mathbf{0} = (0)_{t=1}^n$. The discrete-time Fourier transform (DTFT) of a finite-length sequence $(x_k)_{k=0}^{K-1}$ is the continuous periodic function of frequency $\phi$, defined as $X(\omega) = \sum_{k=0}^{K-1} x_k e^{-j\omega k}$. The identification capacity of a channel is denoted by $\mathbb{C}_{\rm I}$. Throughout this paper, $\bG$ denotes a discrete-time Gaussian channel with ISI and colored noise.

	\subsection{Organization}
	The remainder of this paper is organized as follows. Section~\ref{Sec.SysModel} presents the fundamental prerequisites of the communication systems modeled by the ISI Gaussian channels with colored noise $\bG.$ The main contributions related to $\bG$ are discussed in Section~\ref{Sec.Res}. Section~\ref{Sec.Interpretations} examines the implications of the capacity bounds. Finally, Section~\ref{Sec.Conclusion} concludes the paper and outlines directions for future research.
	
	\section{System Model and Coding Preliminaries}
	\label{Sec.SysModel}
	Here, we introduce the adopted system model and establish preliminaries for coding and capacity.
	
	% ----------------------
	\subsection{Colored Gaussian Channel}
	We consider a channel with $K$-tap ISI and additive colored Gaussian noise with covariance matrix $\fSigma$. The memory is described by a channel impulse response (CIR) sequence $\fh = (h_k)_{k=0}^{K-1}$, where $h_k \in \mathbb{R}$ for all $k$ is known as the CIR tap at time $k\,, \forall k \in [\![K-1]\!]$ with $h_0 h_{K-1} \neq 0$. Let $X_t \in \mathbb{R}$ and $Y_t \in \mathbb{R}$ denote the transmitted and received symbols at time $t$, respectively. The corresponding letter-wise channel law is given by
	% ---
	\begin{align}
		\label{Eq.Law_Letter}
		Y_t = X_t^{\fh} + Z_t,
	\end{align}
	% ---
	where the additive noise affecting the received signal is modeled by the random vector $\fZ,$ which follows a multivariate Gaussian distribution, i.e., $\fZ \sim \N(\f0_{\bar{n}\times\bar{n}},\fSigma_{\bar{n}\times\bar{n}})$ where the covariance matrix $\fSigma \in \mathbb{R}^{\bar{n} \times \bar{n}}$ characterizes the correlation between the noise samples with $\fSigma = [\Sigma_{t,t'}] = \text{Cov}[Z_t,Z_{t'}],\, \forall t,t' \in [\![\bar{n}]\!].$ The multivariate Gaussian distribution density of $\fZ$ reads
	% ---
	\begin{align}
		f_{\fZ}(\fz) & = \frac{1}{\sqrt{(2\pi)^n | \fSigma|}} \exp \Big( - (\fy - \fx^{\fh})^T \fSigma^{-1} (\fy - \fx^{\fh})/2 \Big) ,
	\end{align}
	% ---
	where $|\fSigma|$ is the determinant of $\fSigma.$ Since the channel exhibits dispersion, each output symbol depends on the $K$ most recent input symbols. Consequently, the receiver observes a sequence of length $\bar{n} = n + K - 1$, referred to as the output vector. Hence, based on the conditional distribution of $\bG$ in \eqref{Eq.Law_Letter}, the transition probability distribution can be expressed in the following compact form:
	% ---
	\begin{align}
		\mathbf{Y} = \fH \fx + \mathbf{Z},
	\end{align}
	% ---
	where $\fY$ and $\fZ$ are output and noise vector, i.e., $\fY = (Y_t)_{t=1}^{\bar{n}},$ and $\fZ = (Z_t)_{t=1}^{\bar{n}},$ and $\fH,$ is a full-rank convolution matrix with a Toeplitz structure, where $\fH_{\bar{n} \times n} = [h_{i-j}],$ with $h_k=0$ for $k<0$ or $k \geq K.$ Moreover, setting $\fH \fx = \fx^{\fh},$ we have $f_{\fZ}(\fz) = (2\pi)^{-n/2} | \fSigma|^{-1/2} \exp \big[ - \| \fSigma^{-1/2} ( \fy - \fx^{\fh} ) \|^2 / 2 \big].$ The codewords are subjected to constraint $|x_t| \leq P_{\rm max}, \forall t \in [\![n]\!],$ where $P_{\rm max}> 0$ constrain the per-symbol signal energy and $|x_t|$ is the absolute value of $x_t.$
	
	% ---
	\subsection{Identification Coding}
	In the following, we draw on the rigorous performance parameters for identification established in \cite{A06} and develop a refined and tailored formulation of the code definition and capacity for $\bG.$
	% ---
	\begin{definition}[Colored Gaussian identification code]
		\label{Def.ISI-Poisson-Code}
		An $(n, M(n,R), K(n,\kappa), e_1, e_2)$-DI code for $\bG$ under the peak power constraint $P_{\rm max}$, with integers $M(n,R)$ and $K(n,\kappa)$ and parameters $n$ (codeword length) and $R$ (coding rate), is defined as a system $(\mathbbmss{C}, \D)$ comprising a codebook $\mathbbmss{C} = \{\fc^i\}$ such that
		% ---
		\begin{align}
			\label{Ineq.Constraints}
			- P_{\rm max} \leq c_{i,t} \leq P_{\rm max},
		\end{align}
		% ---
		and a collection of decoding regions $\D = \{ \mathbbmss{D}_i \}, \forall i \in [\![M]\!],\,\forall t \in [\![n]\!].$ Two decoding error events may occur. These events correspond to type I and type II errors, respectively, and are given by
		% ---
		\begin{align}
			\label{Eq.TypeIError}
			P_{e,1}(i) & = \Pr \big( \fY \in \mathbbmss{D}_i^c \,\big|\, \fx = \fc_i \big) = 1 - \int_{\mathbbmss{D}_i} f_{\fZ}(\fy - \fc_i^{\fh}) \, d\fy ,
			\\
			P_{e,2}(i,j)& = \Pr \big( \fY \in \mathbbmss{D}_j \,\big|\, \fx = \fc_i \big) = \int_{\mathbbmss{D}_j} f_{\fZ}(\fy - \fc_i^{\fh}) \, d\fy .
			\label{Eq.TypeIIError}
		\end{align}
		% ---
		It must hold that $P_{e,1}(i) \leq e_1$ and $P_{e,2}(i,j) \leq e_2, \forall \, i,j \in [\![M]\!]$ such that $i \neq j, \allowbreak \, \forall e_1, \allowbreak e_2 \allowbreak > 0.$
		\qed
	\end{definition}
	% ---
	\begin{definition}[Colored Gaussian identification capacity] A rate $R>0$ is said to be DI-achievable if, for any $e_1, e_2>0$ and sufficiently large $n$, there exists an $(n, M(n,R), K(n,\kappa), e_1, e_2)$-DI code. The operational DI capacity of the colored Gaussian channel $\bG$ is then defined as the supremum of all such achievable rates and is denoted by $\mathbb{C}_{\text{I}}(\bG)$.
	\qed
	\end{definition}
	% ---
	
	% ---
	\section{Identification Capacity of the Colored Gaussian Channel with ISI}
	\label{Sec.Res}
	Here, we present our main capacity theorem with the achievability and the converse proofs.
	
	% ---
	\subsection{Main Results}
	\label{Subsec.Main_Results}
	
	First, we introduce a class of CIRs $\fh$ defined through three rigorously specified conditions, each of which include essential criteria for ensuring reliable identification.
	% ---
	\begin{itemize}[leftmargin=*]
		\item \textbf{\textcolor{mycolor12}{C1 (Stability Constraint):}} We assume that the CIR features a finite energy: $\sum_{k=0}^{K-1} |h_k| < \infty,$ which implies: $|h_k| \leq L < \infty\,, \forall k \in [\![K-1]\!].$
		\item \textbf{\textcolor{mycolor12}{C2 (Frequency Spectrum):}} Let $H(\omega)$ be the discrete-time Fourier transform (DTFT) of the CIR vector $\fh.$ Then, we assume that $\inf_{\omega \in [-\pi,\pi]} |H(\omega)| > 0.$
		\item \textbf{\textcolor{mycolor12}{C3 (Covariance Matrix):}} We assume that the covariance matrix $\fSigma$ is full rank and that its singular values of the covariance matrix $\fSigma$ lie in a polynomial range, that is, $\fSigma$ is polynomially well conditioned. More specifically, $\fSigma$ fulfills: $\sigma_{\rm min}(\fSigma) \in \Omega(n^{-\mu})$ and $\sigma_{\rm max}(\fSigma) \in \mathcal{O}(n^{\mu/2}),$
		% $$n^{-\mu} \lesssim \sigma_t(\fSigma) \lesssim n^{\mu/2},$$
	    where $\mu \in [0,1/2)$ is referred to as the spectrum rate.
		% This implies that the convolution in the frequency domain is invertible and numerically stable.
	\end{itemize}
	% ---

	% ---
	\begin{theorem}
		\label{Th.Colored-Capacity}
		Consider the colored Gaussian channel, $\bG,$ with CIR $\fh$ and covariance matrix $\fSigma$ fulfilling conditions \textbf{C1}-\textbf{C3} and assume that the number of ISI channel taps grows sub-linearly with the codeword length, i.e., $K(n,\kappa) = n^{\kappa},$ where $\kappa \in [0,1/2)$ and $\kappa + \mu \in [0,1/2).$ Then, the identification capacity of $\bG$ subject to peak power constraint according to Definition~\ref{Def.ISI-Poisson-Code} and in the super-exponential codebook size scale, i.e., $M(n,R) = 2^{(n\log n)R},$ reads
		% ---
		\begin{align}
			\label{Ineq.LU}
			\frac{1 - 2(\kappa + \mu)}{4} \leq \mathbb{C}_{\rm I}(\bG) \leq 1+\kappa + \frac{\mu}{2} .
		\end{align}
	\end{theorem}
	% ---
	\begin{proof}
		Proofs for achievability and converse are provided in Subsections~\ref{Subsec.Achievability} and \ref{Subsec.Converse}, respectively.
	\end{proof}
	% ---

	In the following, we provide the achievability proof of Theorem~\ref{Th.Colored-Capacity}.
	% ---
	\subsection{Achievability}
	\label{Subsec.Achievability}
	The proof mirrors mainly the same line of construction as of the white Gaussian channel \cite{Salari26}.
	
	\textbf{\textcolor{mycolor12}{Codebook Construction:}} In the following, we deal with an original codebook $\mathbbmss{C} = \{ \fc_i \} \subset \mathbb{R}^n,$ with $i \in [\![M]\!]$  induced by the peak power constraint and an auxiliary codebook referred to as the convoluted codebook denoted by $\mathbbmss{C}^{\fh} = \{ \mathbf{c}_i^{\fh} \} \subset \mathbb{R}^{\bar{n}} ,$ with $i \in [\![M]\!],$ where each $\mathbf{c}_i^{\fh} \triangleq (c_{i,1}^{\fh},\ldots,c_{i,\bar{n}}^{\fh})$ is referred to as a convoluted codeword with
	% ---
	\begin{align}
		\label{Eq.Convoluted-Symbol}
		c_{i,t}^{\fh} \triangleq \sum_{k=0}^{K-1} h_k c_{i,t-k},
	\end{align}
	% ---
	where $c_{i,t} = 0,\, \forall t \leq 0.$ Next, let define the original and the convoluted codebooks as follows:
	% ---
	\begin{align}
		\mathbbmss{C} = \mathbbmss{Q}_{\f0}(n,2P_{\,\text{max}}) & \triangleq \big\{ \fc_i \in \mathbb{R}^n:\; - P_{\rm max} \leq c_{i,t} \leq P_{\rm max} , \forall \, i \in [\![M]\!] , \forall \, t \in [\![n]\!] \big\} ,
		\\
		\mathbbmss{C}^{\fh} & \triangleq \big\{ \fc_i^{\fh} \in \mathbb{R}^{\bar{n}} :\, c_{i,t}^{\fh} \triangleq \sum_{k=0}^{K-1} h_k c_{i,t-k}:\, \fc_i \in \mathbbmss{C}, \forall \, i \in [\![M]\!] \big\} .
	\end{align}
	% ---
	
	\begin{customlemma}{1}[minimum distance of the convoluted codebook]
		\label{Lem.Min_dis_con_cb}
		Let $H(\omega)$ denote the DTFT of the CIR vector corresponding to $\bG$. Then, the minimum distance of the convolved codebook $\mathbbmss{C}^{\fh}$ satisfies:
		% ---
		\begin{align}
			\label{Ineq.LB_Norm_Conv_CB}
			\| \fc_i^{\fh} - \fc_j^{\fh} \| \geq H_{\rm min} \| \fc_i - \fc_j \|,
		\end{align}
		% ---
		where $H_{\rm min} \triangleq \inf_{\omega \in [0,2\pi]} |H(\omega)|/2\pi.$
		% ---
		\begin{proof}
			The proof provided in the proof of \cite[Lem. 1]{Salari26}.
		\end{proof}
		% ---
	\end{customlemma}
	% ---
	
	\textbf{\textcolor{mycolor12}{Rate Analysis:}}
	We use a packing arrangement of non-overlapping hyper spheres of radius $r_0 = \sqrt{\bar{n}\epsilon_n}$ in a hyper cube with edge length $P_{\rm max},$ where
	% ---
	\begin{align}
		\label{Eq.Epsilon_n}
		\epsilon_n = \frac{a}{H_{\rm min}^2n^{( 1 - (2\kappa + 2\mu + b))) / 2}},
	\end{align}
	% ---
	with $a > 0$ being a fixed constant and $b$ denoting an arbitrarily small constant.
	
	Let $\mathscr{S}$ denote a sphere packing, i.e., an arrangement of $M$ non-overlapping spheres $\S_{\fc_i}(n,r_0),\, i \in [\![M]\!],$ that are packed inside the larger cube $\mathbbmss{Q}_{\f0}(n,P_{\rm max}).$ Following the same approach as presented for the white Gaussian channel \cite{Salari26} we conform to a relaxed geometric structure, we require only that the centers of the spheres lie within the hypercube $\mathbbmss{Q}_{\f0}(n,P_{\rm max}),$ that the spheres are mutually disjoint, and that each sphere exhibits a non-empty intersection with $\mathbbmss{Q}_{\f0}(n,P_{\rm max}).$ The packing density \cite{CHSN13} is
	% ---
	\begin{align}
		\Updelta_n(\mathscr{S}) \triangleq \frac{\text{Vol}\Big(\mathbbmss{Q}_{\f0}(n,P_{\rm max}) \cap \bigcup_{i=1}^{M}\S_{\fc_i}(n,r_0) \Big)}{\text{Vol}\big(\mathbbmss{Q}_{\f0}(n,P_{\rm max}) \big)}.
		\label{Eq.Def_Density}
	\end{align}
	% ---
	% ---
	We invoke a saturated packing argument as accomplished in \cite{Salari26}. Specifically, consider a saturated packing of spheres \(\bigcup_{i=1}^{M(n,R)} \S_{\fc_i}(n,r_0)\) with radius \( r_0 = \sqrt{\bar{n}\epsilon_n} \), embedded within the hypercube \(\mathbbmss{Q}_{\f0}(n,P_{\rm max})\). In general, the volume of a hypersphere of radius \( r \) is given by \cite[Eq.~(16)]{CHSN13},
	% ---
	\begin{align}
		\text{Vol}\big(\S_{\fc_i}(n,r)\big) = \frac{\pi^{\frac{n}{2}}}{\Gamma(\frac{n}{2}+1)} \cdot r^{n} .
		\label{Eq.VolS}
	\end{align}
	% ---
	Note that the density of such an arrangement fulfills \cite[Sec.~IV]{Salariseddigh-TMBMC}
	% ---
	\begin{align}
		\label{Ineq.Density}
		2^{-n} \leq \Updelta_n(\mathscr{S}) \leq 2^{-0.599n} .
	\end{align}
	% ---
	We associate each hypersphere with a codeword located at its center $\mathbf{c}_i$, where $\|\mathbf{c}_i\|_{\infty} \leq P_{\mathrm{max}}$. Given that each sphere has volume $\mathrm{Vol}(\mathcal{S}_{\mathbf{c}_1}(n,r_0))$ and all centers lie within $\mathbbmss{Q}_{\mathbf{0}}(n,P_{\mathrm{max}})$, the number of packed spheres, $M$, reads
	% ---
	\begin{align}
		\label{Eq.M}
		M = \frac{\text{Vol}\big(\bigcup_{i=1}^{M}\S_{\fc_i}(n,r_0)\big)}{\text{Vol}(\S_{\fc_1}(n,r_0))} & \geq \frac{\text{Vol}\big(\mathbbmss{Q}_{\f0}(n,P_{\rm max}) \cap \bigcup_{i=1}^{M}\S_{\fc_i}(n,r_0) \big)}{\text{Vol}(\S_{\fc_1}(n,r_0))}
		\stackrel{(a)}{\geq} \frac{(P_{\rm max}/2)^n}{\text{Vol}(\S_{\fc_1}(n,r_0))} ,
	\end{align}
	% ---
	where $(a)$ exploits \eqref{Eq.Def_Density} and \eqref{Ineq.Density}. The bound in \eqref{Eq.M} admits the following simplification
	% ---
	\begin{align}
		\label{Eq.Log_M_0}
		\log M \stackrel{(a)}{\geq} n \log P_{\rm max} - n \log r_0 + \floor{n/2} \log \floor{n/2} - \floor{n/2} \log e + o \big( \floor{n/2} \big) - n ,
	\end{align}
	% ---
	where $(a)$ uses \eqref{Eq.VolS} and Stirling's approximation, namely, $\log n! = n \log n - n \log e + o(n)$ \cite[P.~52]{F66} with setting $n$ with $\floor{n/2} \in \mathbb{Z},$ and since $\Gamma (( n/2) + 1 ) \geq \floor{n/2} !$ see \cite{Salari26} for details. Now, observe
	% ---
	\begin{align}
		r_0 = \allowbreak \sqrt{\bar{n}\epsilon_n} \sim \sqrt{n\epsilon_n} = \sqrt{a}H_{\rm min}^{-1}n^{(1 + 2\kappa + 2\mu + b) / 4}.
	\end{align}
	% ---
	Accordingly, we arrive at the following bound on the logarithm of $M,$
	% ---
	\begin{align}
		\log M  \geq \left( \frac{2 - (1 + 2\kappa + 2\mu + b)}{4} \right) n \log n + n \log \Big( \frac{P_{\rm max}H_{\rm min}}{\sqrt{ae}} \Big)  + \mathcal{O}(n) ,
		\label{Eq.Log_M}
	\end{align}
	% ---
	see \cite{Salari26} for detailed derivations. Consequently, the leading-order term in \eqref{Eq.Log_M} is of order $n \log n$. Ensuring that the derived lower bound on the achievable rate, $R,$ remains finite as $n \to \infty,$ requires a corresponding scaling of $M.$ In particular, $M$ must scale as $M = 2^{(n \log n)R}.$ Therefore,
	
	% ---
	\begin{align}
		R \geq \frac{1}{n \log n} \left[ \left( \frac{2 - (1 + 2\kappa + 2\mu+ b)}{4} \right) n \log n + n \log \Big( \frac{P_{\rm max}H_{\rm min}}{\sqrt{ae}} \Big) + o(n \log n) \right] ,
	\end{align}
	% ---
	which tends to $(1-2(\kappa + \mu))/4$ when $n \to \infty$ and $b \rightarrow 0.$
	% ---

	\textbf{\textcolor{mycolor12}{Encoding:}} We assume that the encoding function is deterministic, i.e., each message $i \in [\![M]\!]$ is associated to a known codeword $\fc_i.$ Hence, given $i \in [\![M]\!],$ the transmitter sends $\fx = \fc_i.$
	
	\textbf{\textcolor{mycolor12}{Decoding:}}
	Let $e_1, e_2, \eta_0, \zeta_0, \zeta_1 > 0$ be arbitrarily small constants. Before proceeding, we set the following conventions to ensure a clear and focused analysis:
	% ---
	\begin{itemize}[leftmargin=*]
		\item $Y_t(i) = c_{i,t}^{\fh} + Z_t,\, \forall t \in [\![\bar{n}]\!]$ denotes the channel output at time $t$ \emph{conditioned} that $\fx=\fc_i$ was sent.
		\item $\fZ = \fY(i) - \fc_j^{\fh}$ denotes the colored noise vector.
		\item $\fZ_{\rm w} \triangleq \fSigma^{-1/2} \fZ$ denotes the whitened noise vector.
		\item The output vector consists of the symbols, i.e., $\fY(i)= (Y_1(i),\ldots, Y_{\bar{n}}(i))$ with $\bar{n} = n+K-1.$
		\item $c_{i,t}^{\fh} \triangleq \sum_{k=0}^{K-1} h_k c_{i,t-k}$ is the convoluted symbol, i.e., the linear combination of $\fc_i$ and $\fh.$
		\item $\delta_n \hspace{-.4mm} = \hspace{-.4mm} 4aC_{\sigma_{\rm max}} / 3n^{(1 - (2\kappa + \mu + b))/2}$ is \emph{decoding threshold} with $a,b>0$ being fixed and arbitrary constants.
		\item The frequency response is bounded away from zero over its support: $H_{\rm min} \triangleq \inf_{\omega \in [0,2\pi]} |H(\omega)| > 0.$
	\end{itemize}
	% ---
	To determine if message $j \in [\![M]\!]$ was sent, the decoder checks if $\mathbf{y}$ lies in the decoding set:
	% ---
	\begin{align}
		\label{Eq.Dec_Meas0}
		\mathbbmss{D}_j = \Big\{ \fy \in \mathbb{R}^{\bar{n}} \,:\; |T(\fy,\fc_j^{\fh})| \leq \delta_n \Big\},
	\end{align}
	% ---
	with $T(\fy,\fc_j^{\fh}) = \bar{n}^{-1} (\fy - \fc_j^{\fh})^T \fSigma^{-1} (\fy - \fc_j^{\fh}) - 1$ being referred to as the \emph{decoding measure} where
	% ---
	\begin{align}
		\label{Eq.d_mah}
		\bar{n}^{-1} \big( (\fy - \fc_j^{\fh})^T \fSigma^{-1} (\fy - \fc_j^{\fh}) \big) ,
	\end{align}
	% ---
	% ---
	%\begin{align}
	%	\label{Eq.Dec_Meas1}
	%	T(\fy,\fc_j^{\fh}) = \bar{n}^{-1} \big( (\fy - \fc_j^{\fh})^T \fSigma^{-1} (\fy - \fc_j^{\fh}) \big) - 1 = \bar{n}^{-1} \| \fSigma^{-1/2} ( \fy - \fc_j^{\fh} ) \| - 1 ,
	%\end{align}
	% ---
	% $T(\fy,\fc_j^{\fh}) = \bar{n}^{-1} \big\| \fy -  \fc_j^{\fh} \big\|^2 - \bar{n}^{-1} \sum_{t=1}^{\bar{n}} \sigma_{Z_t}^2$ 
is the normalized squared \emph{Mahalanobis distance} \cite{Mahalanobis18} between output $\fy$ and its mean $\fc_j$ under $f_{\fZ}(\fz)$.
	%$\text{tr}(\fSigma) = \sum_{t=1}^{\bar{n}} \sigma_{Z_t}^2$ being the trace of the covariance matrix.
	
	To simplify notation, we adopt the following definitions throughout the error analysis:
	% ---
	\begin{itemize}[leftmargin=*]
		\item $\fd_{j} \triangleq \| \fZ_{\rm w} \| = \| \fSigma^{-1/2} ( \fY(i) - \fc_j^{\fh} ) \|.$
		\item $T(\fY(i),\fc_j^{\fh}) = \bar{\fd}_j^2 - 1$ with $\bar{\fd}_j^2 \triangleq \bar{n}^{-1} \fd_{j}^T \fd_{j}  = \bar{n}^{-1}  \| \fZ_{\rm w} \|^2 = \bar{n}^{-1} (\fY(i) - \fc_j^{\fh})^T \fSigma^{-1} (\fY(i) - \fc_j^{\fh}).$
		\item $\fd_{i,j} \triangleq \fSigma^{-1/2} ( \fc_i^{\fh} - \fc_j^{\fh} )$ and $\ff_{i,j} \triangleq \fSigma^{-1/2} \fd_{i,j}.$
		\item $U_{i,j} \triangleq \bar{n}^{-1} \big(\big\| \fZ_{\rm w} \big\|^2 + \big\| \fd_{i,j}  \big\|^2 \big).$
		\item $V_{i,j} \triangleq 2\bar{n}^{-1} \fZ_{\rm w}^T \fSigma^{-1} ( \fc_i^{\fh} - \fc_j^{\fh} ).$
		\item $W_{i,j} \triangleq U_{i,j} + V_{i,j}.$
		\item $\mathbbmss{E}_0 \triangleq \{ |V_{i,j}| > \delta_n \} = \big\{ \fZ \in \mathbb{R}^{\bar{n}} \;:\, \big|  2\bar{n}^{-1} \fZ_{\rm w}^T \fSigma^{-1} ( \fc_i^{\fh} - \fc_j^{\fh} ) \big| > \delta_n \big\}.$
		\item $\mathbbmss{E}_1 \triangleq \{ U_{i,j} - 1 \leq 2\delta_n \} = \big\{ \fZ \in \mathbb{R}^{\bar{n}} \;:\, \bar{n}^{-1} \big( \big\| \fZ_{\rm w} \big\|^2 + \big\| \fd_{i,j}  \big\|^2 \big) - 1 \leq 2\delta_n \big) \big\}.$
		\item $\mathbbmss{E}_2 \triangleq \{ W_{i,j} - 1 \leq \delta_n \} = \big\{ \fZ \in \mathbb{R}^{\bar{n}} \;:\, \bar{n}^{-1} \big\| \fZ_{\rm w} + \fd_{i,j} \big\|^2 - 1 \leq \delta_n \big\}.$
	\end{itemize}
	% ---
	\textbf{\textcolor{mycolor12}{Type I:}}
	The type I errors occur when the transmitter sends $\fc^i,$ yet $\fY \notin \mathbbmss{D}_i.$ For every $i \in [\![M]\!],$ the type I error probability is bounded by 
	% ---
	\begin{align}
		\label{Eq.TypeIError}
		P_{e,1}(i) = \Pr\big( \fY(i) \in \mathbbmss{D}_i^c \big) = \Pr\big( T(\fY(i),\fc_i^{\fh}) > \delta_n \big).
	\end{align}
	% ---
	To bound $P_{e,1}(i),$ we perform Chebyshev's inequality, namely,
	% ---
	\begin{align}
		\label{Ineq.TypeI_Cheb}
		\Pr\big( \big| T(\fY(i),\fc_i^{\fh}) - \mathbb{E} \big[ T(\fY(i),\fc_i^{\fh}) \big] \big| > \delta_n \big) \leq \frac{\text{Var} \big[ T(\fY(i),\fc_i^{\fh}) \big]}{\delta_n^2} .
	\end{align}
	% ---
	Next, to calculate the expectation of the decoding measure, we exploit a helpful lemma.
	% ---
	\begin{customlemma}{2}
		\label{Lem.WNA}
		The squared Mahalanobis distance $(\fY - \fx^{\fh})^T \fSigma^{-1} (\fY - \fx^{\fh})$ follows a chi-squared distribution \cite{Papoulis02} with $\bar{n}$ degree of freedom, i.e., $(\fY - \fx^{\fh})^T \fSigma^{-1} (\fY - \fx^{\fh}) \sim \chi_{\bar{n}}^2.$
	\end{customlemma}
	% ---
		% ---
	\begin{proof}
		The proof is provided in Appendix \ref{App.WNA}.
	\end{proof}
	% ---
	
	Now, we start to calculate the expectation of the decoding measure as follows
	% ---
	\begin{align}
		\label{Ineq.Exp_Decoding_Metric}
		\mathbb{E} \big[ T( \fY(i),\fc_i^{\fh}) \big] &	\stackrel{(a)}{=} \bar{n}^{-1} \mathbb{E} [  \| \fZ_{\rm w} \|^2 ] - 1
		\stackrel{(b)}{=} \bar{n}^{-1}  \sum_{t=1}^{\bar{n}} \mathbb{E}[Z_{\rm w,t}^2] - 1
		\stackrel{(c)}{=} \bar{n}^{-1}  \sum_{t=1}^{\bar{n}} 1 - 1 = 0,
	\end{align}
	% ---
	where $(a)$ uses Lemma \ref{Lem.WNA}, with setting $\fY = \fY(i)$ and $\fx = \fc_j,$ i.e., $\fd_{j}^2 \sim \chi_{\bar{n}}^2,$ $(b)$ uses the linearity of the expectation and $(c)$ exploits $Z_{\rm w,t} \overset{\text{\tiny i.i.d}}{\sim} \N(0,1)$ with $\text{Var}[Z_{\rm w,t}] =  \mathbb{E}[Z_{\rm w,t}^2] = 1.$ Second, the variance of the decoding measure is given by
	% --- $\text{Var}[\fZ_{\rm w}] = \mathbb{E}[\fZ_{\rm w}^2] = \sum_{t=1}^{\bar{n}} \mathbb{E}[Z_{\rm w,t}^2] = \sum_{t=1}^{\bar{n}} 1 = \bar{n}.$
	\begin{align}
		\label{Ineq.Var_Decoding_Metric}
		\text{Var}\big[ T(\fY(i),\fc_i^{\fh}) \big] = \bar{n}^{-2} \text{Var}[\| \fZ_{\rm w} \|^2 ] \stackrel{(a)}{=} \bar{n}^{-2} \sum_{t=1}^{\bar{n}} \text{Var} [Z_{\rm w,t}^2] 
		& \stackrel{(b)}{=} \bar{n}^{-2} \big( \sum_{t=1}^{\bar{n}} 3\sigma_{Z_{\rm w,t}}^4 - \sigma_{Z_{\rm w,t}}^2 \big) = 2\bar{n}^{-1},
	\end{align}
	% ---
	where $(a)$ invokes $Z_{\rm w,t} \overset{\text{\tiny i.i.d}}{\sim} \N(0,1)$ and $(b)$ holds since $\text{Var}[Z_{\rm w,t}^2] = \mathbb{E}[Z_{\rm w,t}^4] - (\mathbb{E}[Z_{\rm w,t}^2])^2 $ and $\mathbb{E}[Z_t^4] = 3\sigma_{Z_t}^4$ for $Z_t \overset{\text{\tiny i.i.d}}{\sim} \N(0,\sigma_{Z_t}^2)$ with setting $Z_t = Z_{\rm w,t}.$
	% $$\text{Var}\Big(\sum_{t=1}^{\bar{n}} X_t \Big) = \sum_{t=1}^{\bar{n}} \sum_{t'=1}^{\bar{n}} \text{Cov}(X_t,X_{t'})$$ $\text{Cov}[X,Y] \leq \sqrt{\text{Var}[X] \cdot \text{Var}[Y]}$ for RVs with finite variances	
	Thereby, employing \eqref{Ineq.Exp_Decoding_Metric} and \eqref{Ineq.Var_Decoding_Metric} into \eqref{Ineq.TypeI_Cheb} yields
	% ---
	\begin{align}
		\label{Ineq.TypeI_Final}
		P_{e,1}(i) & \stackrel{(a)}{\leq} \frac{\text{Var} \big[ T(\fY(i),\fc_i^{\fh}) \big]}{\delta_n^2}
		\leq \frac{2}{\bar{n}\delta_n^2}\stackrel{(b)}{\leq} \frac{9}{8a^2C_{\sigma_{\rm max}}^2 n^{2\kappa + \mu + b}} \triangleq \eta_0,
	\end{align}
	% ---
	where $(a)$ employs the Chebyshev's inequality and $(b)$ uses $\delta_n = 4aC_{\sigma_{\rm max}} / 3n^{(1 - (2\kappa + \mu + b))/2}$ and $n \leq \bar{n}.$ Hence, $P_{e,1}(i) \leq \eta_0 \leq e_1$ holds for sufficiently large $n$ and arbitrarily small $e_1 > 0.$
	
	\textbf{\textcolor{mycolor12}{Type II:}}
	We examine type II errors, i.e., when $\fY \in \mathbbmss{D}_j$ while the transmitter sent $\fc_i$ with $i \neq j \,.$ Then, for every $i,j \in [\![M]\!],$ the type II error probability is given by
	% --- 
	\begin{align}
		P_{e,2}(i,j) = \Pr \big( \big| T(\fY(i);\fc_j^{\fh}) \big| \leq \delta_n \big).
		\label{Eq.Pe2G}
	\end{align}
	% ---
	Next, exploiting the reverse triangle inequality, i.e., $|W_{i,j}| - |1| \leq |W_{i,j} - 1|,$ we obtain
	% ---
	\begin{align}
		& P_{e,2}(i,j) \leq \Pr\big( |W_{i,j}| - |1| \leq \delta_n \big)
		\stackrel{(a)}{=} \Pr\big( W_{i,j} - 1 \leq \delta_n \big) \stackrel{(b)}{=} \Pr(\mathbbmss{E}_2),
	\end{align}
	% ---
	where $(a)$ follows since $W_{i,j} \geq 0,$ and $(b)$ holds by the following argument:
	% ---
	\begin{align}
		\fd_{i,j}^2 = (\fY(i) - \fc_j^{\fh})^T \fSigma^{-1} (\fY(i) - \fc_j^{\fh}) & = (\fY(i) - \fc_i^{\fh} + \fc_i^{\fh} - \fc_j^{\fh})^T \fSigma^{-1} (\fY(i)  - \fc_i^{\fh} + \fc_i^{\fh} - \fc_j^{\fh})
		\nonumber\\&
		= \big( \fSigma^{-1/2} (\fZ + \fc_i^{\fh} - \fc_j^{\fh}) \big)^T \big( \fSigma^{-1/2} (\fZ + \fc_i^{\fh} - \fc_j^{\fh}) \big)
		\nonumber\\&
		= \| \fSigma^{-1/2} (\fZ + \fc_i^{\fh} - \fc_j^{\fh}) \|^2 = \| \fZ_{\rm w} + \fd_{i,j} \|^2 .
		% \big\| \fY(i) - \fc_j^{\fh} \big\|^2 = \big\| \fY(i) - \fc_i^{\fh} + \fc_i^{\fh} - \fc_j^{\fh} \big\|^2 = \big\| \fZ + \fc_i^{\fh} - \fc_j^{\fh} \big\|^2.
	\end{align}
	% ---
	Next, in order to bound the event $\mathbbmss{E}_2$ we employ $\| \fA \fx \|^2 = (\fA \fx)^T (\fA \fx)$ where $\fA$ is a matrix and $\fx$ is a vector, and decompose the square norm given in the event $\mathbbmss{E}_2$ as follows
	% ---
	\begin{align}
		\big\| \fZ_{\rm w} + \fd_{i,j} \big\|^2 = ( \fZ_{\rm w} + \fd_{i,j} )^T( \fZ_{\rm w} + \fd_{i,j} ) & = \fZ_{\rm w}^T \fZ_{\rm w} + \fd_{i,j}^T \fd_{i,j} + 2 \fZ_{\rm w}^T \fd_{i,j}
		\nonumber\\&
		= \big\| \fZ_{\rm w} \big\|^2 + \big\| \fd_{i,j}  \big\|^2 + 2 \fZ_{\rm w}^T \fSigma^{-1/2} ( \fc_i^{\fh} - \fc_j^{\fh} ) .
		%\nonumber\\&
		%= \bar{n}^{-1} \big( \big\| \fZ_{\rm w} \big\|^2 + \big\| \fd_{i,j}  \big\|^2 \big) + 2\bar{n}^{-1} \fZ^T \fSigma^{-1} ( \fc_i^{\fh} - \fc_j^{\fh} )
		%	\stackrel{(a)}{=} \bar{n}^{-1} \big( \big\| \fSigma^{-1/2} \fZ \big\|^2 + \big\| \fSigma^{-1/2} ( \fc_i^{\fh} - \fc_j^{\fh} ) \big\|^2 \big) + 2\bar{n}^{-1} \sum_{t=1}^{\bar{n}} \sum_{t'=1}^{\bar{n}} Z_{t} \nu_{t,t'} ( c_{i,t}^{\fh} - c_{j,t}^{\fh} )
		\label{Eq.TypeII-3}
	\end{align}
	% ---
	%where $(a)$ expands the summation in the cross-product tern over symbols.
	Next, we establish the variance for the cross-product in \eqref{Eq.TypeII-3} as follow
	% ---
	\begin{align}
		\label{Eq.Var_Cross_Prod}
		\text{Var} \big[ \fZ_{\rm w}^T \fSigma^{-1/2} ( \fc_i^{\fh} - \fc_j^{\fh} ) \big] & \stackrel{(a)}{=} \text{Var} \big[ (\fSigma^{-1/2} \fZ )^T \fSigma^{-1/2} ( \fc_i^{\fh} - \fc_j^{\fh} ) \big]
		\nonumber\\&
		\stackrel{(b)}{=} \text{Var} \big[ \ff_{i,j}^T \fZ \big]
		\nonumber\\&
		\stackrel{(c)}{=} \mathbb{E}\big[ \big( \ff_{i,j}^T \fZ - \ff_{i,j}^T \mathbb{E}[  \fZ] \big)^2 \big] = \mathbb{E}\big[ \big( \ff_{i,j}^T \big( \fZ - \mathbb{E}[  \fZ] \big) \big)^2 \big]
		\nonumber\\&
		\stackrel{(d)}{=} \mathbb{E}\big[ \big( \ff_{i,j}^T ( \fZ - \mathbb{E}[ \fZ] ) \cdot \big( \fZ - \mathbb{E}[ \fZ] \big)^T \ff_{i,j} \big]
		\nonumber\\&
		%\stackrel{(e)}{=}\ff_{i,j}^T \mathbb{E} \big[ ( \fZ - \mathbb{E}[ \fZ] ) \cdot \big( \fZ - \mathbb{E}[ \fZ] \big)^T \big] \ff_{i,j}
		%\nonumber\\&
		\stackrel{(e)}{=} \ff_{i,j}^T \text{Cov}[\fZ] \ff_{i,j}
		= \big( \fSigma^{-1} ( \fc_i^{\fh} - \fc_j^{\fh} ) \big)^T \fSigma  \big( \fSigma^{-1} ( \fc_i^{\fh} - \fc_j^{\fh} ) \big)
		\nonumber\\&
		\stackrel{(f)}{=} ( \fc_i^{\fh} - \fc_j^{\fh} )^T \fSigma^{-1} ( \fc_i^{\fh} - \fc_j^{\fh} ) ,
	\end{align}
	% ---
	where we use the followings
	% ---
	\begin{itemize}
		\item $(a)$ uses $\fZ_{\rm w} = \fSigma^{-1/2} \fZ.$
		\item $(b)$ invokes $\ff_{i,j} \triangleq \fSigma^{-1} ( \fc_i^{\fh} - \fc_j^{\fh} )$ and $\text{Var}[\fZ^T \ff_{i,j}] = \text{Var}[\ff_{i,j}^T \fZ],$ since $\fZ^T \ff_{i,j}$ and $\ff_{i,j}^T \fZ$ are scalars.
		\item $(c)$ uses $\text{Var}[X] = \mathbb{E}[(X - \mathbb{E}[X])^2]$ with setting $X = \ff_{i,j}^T \fZ.$
		\item $(d)$ follows since $X^2 = X \cdot X $ with $X = \ff_{i,j}^T \big( \fZ - \mathbb{E}[  \fZ] \big)$ being a scalar.
		\item $(e)$ holds since $ \ff_{i,j}$ is deterministic and $\text{Cov}[\fZ] = \mathbb{E}[(\fZ-\mathbb{E}[\fX])(\fZ-\mathbb{E}[\fZ])^T] = \fSigma.$
		\item $(f)$ follows by symmetry of the inverse matrix, i.e., $(\fSigma^{-1})^T = \fSigma^{-1}.$
	\end{itemize}
	% ---
	Next, to bound the expression in \eqref{Eq.Var_Cross_Prod}, we employ two helpful lemmas which characterize bounds on the singular values of inverse covariance matrix and the Rayleigh quotient of a matrix, respectively.
	% ---
	
	% ---
	\begin{customlemma}{3}[Upper Bound on The Rayleigh Quotient]
		\label{Lem.RQ}
		Let $\fA \in \mathbb{R}^{n \times n}$ be a symmetric matrix and define the Rayleigh quotient $\forall \fx \in \mathbb{R}^n, \; \fx \neq \f0$ by $R(\fx) = \fx^T \fA \fx / \fx^T \fx.$ Then, it holds that
		% ---
		\begin{align}
			\lambda_{\min}(\fA) \le R(\fx) \le \lambda_{\max}(\fA)
		\end{align}
		% ---
		where $\lambda_{\min}(\fA)$ and $\lambda_{\max}(\fA)$ are the smallest and largest singular values of $\fA.$
	\end{customlemma}
	% ---
	
	% ---
	\begin{proof}
		The proof is provided in Appendix \ref{App.RQ}.
	\end{proof}
	% ---

	\begin{customlemma}{4}[Bounds on the Singular Values of the Inverse Matrix]
		\label{Lem.DB_Sing_Val}
		Let $\fSigma \in \mathbb{R}^{\bar{n} \times \bar{n}}$ be invertible with singular values $\sigma_1(\fSigma) \ge \sigma_2(\fSigma) \ge \cdots \ge \sigma_{\bar{n}}(\fSigma) > 0.$ Then the singular values of $\fSigma^{-1}$ read $\sigma_k(\fA^{-1}) = \sigma_{n-k+1}^{-1}(\fA),\, \forall t \in [\![\bar{n}]\!]$ and in particular 
		% ---
		\begin{align}
			\sigma_1^{-1}(\fSigma) \leq \sigma_t(\fSigma^{-1}) \leq \sigma_{\bar{n}}^{-1}(\fSigma).
		\end{align}
		% ---
	\end{customlemma}
	% ---
	
	% ---
	\begin{proof}
		The proof is provided in Appendix \ref{App.B_Sing_Inv}.
	\end{proof}
	% ---
	% ---
	We now apply the Chebyshev's inequality and exploits Lemma \ref{Lem.RQ} to bound $\Pr (\mathbbmss{E}_0 )$ as follows:
	% ---
	\begin{align}
		\label{Ineq.Event_E0_1}
		\Pr(\mathbbmss{E}_0) \leq \frac{\text{Var} \big[ 2 \bar{n}^{-1} \fZ_{\rm w}^T \fSigma^{-1} ( \fc_i^{\fh} - \fc_j^{\fh} ) \big]}{\delta_n^2} & = \frac{4 \text{Var} \big[\bar{n}^{-1} \fZ_{\rm w}^T \fSigma^{-1} ( \fc_i^{\fh} - \fc_j^{\fh} ) \big]}{\bar{n}^{2}\delta_n^2}
		\nonumber\\
		& \stackrel{(a)}{\le} \frac{4 \sigma_{\rm max}(\fSigma^{-1}) ( \fc_i^{\fh} - \fc_j^{\fh} )^T ( \fc_i^{\fh} - \fc_j^{\fh} ) }{\bar{n}^2\delta_n^2}
		\nonumber\\
		& = \frac{4\sigma_{\rm max}(\fSigma^{-1}) \| \fc_i^{\fh} - \fc_j^{\fh} \|^2}{\bar{n}^2\delta_n^2},
	\end{align}
	% ---
	where $(a)$ employs Lemma \ref{Lem.RQ} with setting $\fA = \fSigma^{-1}$ to upper bound the variance of cross-product term in \eqref{Eq.Var_Cross_Prod} and since for the symmetric positive definite matrix $\fSigma^{-1}$ the singular values and eigenvalues are identical. Now observe that
	% ---
	\begin{align}
		\label{Ineq.Norm_Diff_Squa_UB}
		\big\| \fc_i^{\fh} - \fc_j^{\fh} \big\|^2 & \stackrel{(a)}{\leq} \big( \sqrt{\bar{n}} \big\| \fc_i^{\fh} \big\|_{\infty} + \sqrt{\bar{n}} \big\| \fc_j^{\fh} \big\|_{\infty} \big)^2 = 4\bar{n}K^2L^2P_{\rm max}^2 ,
	\end{align}
	% ---
	where $(a)$ holds by the triangle inequality. Thereby,
	% ---
	\begin{align}
		\label{Ineq.Event_E0_2}
		\Pr(\mathbbmss{E}_0) \stackrel{(a)}{\leq} \frac{4 \| \fc_i^{\fh} - \fc_j^{\fh} \|^2}{\sigma_{\rm min}(\fSigma) \bar{n}^2 \delta_n^2}
		\stackrel{(b)}{\leq} \frac{9a n^{2\kappa}L^2P_{\rm max}^2 (C_{\sigma_{\rm min}}C_{\sigma_{\rm max}} )^{-1} }{n^{-\mu} n^{2\kappa + \mu + b}} = \frac{9a L^2P_{\rm max}^2 (C_{\sigma_{\rm min}}C_{\sigma_{\rm max}} )^{-1} }{n^b} 
		\triangleq \zeta_0,
	\end{align}
	% ---
	where
	% ---
	\begin{itemize}
		\item $(a)$ employs Lemma \ref{Lem.DB_Sing_Val} and the inversion property of singular values, $\sigma_{\min}(\fSigma^{-1})=\sigma_{\max}^{-1}(\fSigma)$.
		\item $(b)$ uses \eqref{Ineq.Norm_Diff_Squa_UB}, $\delta_n = 4aC_{\sigma_{\rm max}} / 3n^{(1 - (2\kappa + \mu + b))/2}$ and $\sigma_{\rm min}(\fSigma) \in \Omega(n^{-\mu})$ with $C_{\sigma_{\rm min}} >0$ and $n \leq \bar{n}.$
	\end{itemize}
	% ---
	 Now, the complementary event $\mathbbmss{E}_0^c,$ gives
	% ---
	\begin{align}
		\label{Eq.E_0_Comp}
		2 \bar{n}^{-1} \fZ_{\rm w}^T \fSigma^{-1} ( \fc_i^{\fh} - \fc_j^{\fh} ) > - \delta_n.
	\end{align}
	% ---
	Next, applying the law of total probability to the event $\mathbbmss{E}_2$ over $\mathbbmss{E}_0$ and its complement $\mathbbmss{E}_0^c,$ gives
	% ---
	\begin{align}
		P_{e,2} (i,j) \leq \Pr(\mathbbmss{E}_2) \stackrel{(a)}{\leq} \Pr(\mathbbmss{E}_0) + \Pr\left( \mathbbmss{E}_2 \cap\,{\mathbbmss{E}_0^c} \right)
		\stackrel{(b)}{\leq} \Pr\left( \mathbbmss{E}_0 \right) + \Pr\left( \mathbbmss{E}_1 \right), \hspace{-1mm}
		\label{Eq.TypeIIError-E_0+E_1} 
	\end{align}
	% ---
	where $(a)$ uses $\mathbbmss{E}_2 \cap \mathbbmss{E}_0 \subset \mathbbmss{E}_0$ and $(b)$ holds by $\Pr(\mathbbmss{E}_2 \cap \mathbbmss{E}_0^c) \leq \Pr (\mathbbmss{E}_1 )$ which is proved in the following:
	% ---
	\begin{align}
		\hspace{-2mm} \Pr(\mathbbmss{E}_2 \cap \mathbbmss{E}_0^c) \hspace{-.4mm} = \hspace{-.4mm} \Pr \big( \big\{ U_{i,j} - 1 \leq \delta_n-V_{i,j} \big\}
		\cap\big\{| V_{i,j} | \leq \delta_n \, \big\}  \big)
		\overset{(a)}{\leq} \Pr \big( \big\{ U_{i,j} - 1 \leq 2 \delta_n  \big\}  \big)
		= \Pr\left(\mathbbmss{E}_1 \right), \hspace{-1mm}
	\end{align}
	% ---
	where $(a)$ uses $\delta_n-V_{i,j}\leq 2\delta_n$ conditioned on $| V_{i,j} | \leq \delta_n.$ We now bound $\Pr\left(\mathbbmss{E}_1 \right).$ Observe that
	% ---
	\begin{align}
		\label{Ineq.Radius_Event}
		\| \fd_{i,j} \|^2 \triangleq \big\| \fSigma^{-1} ( \fc_i^{\fh} - \fc_j^{\fh} ) \big\|^2	\stackrel{(a)}{\geq} \sigma_{\rm max}^{-1}(\fSigma) \big\| \fc_i^{\fh} - \fc_j^{\fh} \big\|^2 \stackrel{(b)}{\geq} 4C_{\sigma_{\rm max}} H_{\rm min}^2\bar{n}\epsilon_n n^{-\mu/2} ,
	\end{align}
	% ---
	where
	% ---
	\begin{itemize}
		\item $(a)$ holds by \cite[Lem. 5]{Salariseddigh_affine_23_arXiv} and since $\sigma_{\rm min}(\fSigma^{-1}) = \sigma_{\rm max}^{-1}(\fSigma);$ see \cite[Ch. 9]{Kunze71}.
		\item $(b)$ employs Lemma \ref{Lem.Min_dis_con_cb} with $\| \fc_i - \fc_j \| \geq 2r_0 = 2\sqrt{\bar{n}\epsilon_n},$ and $\sigma_{\rm max}(\fSigma) \in \mathcal{O}(n^{\mu/2})$ with $C_{\sigma_{\rm max}}>0.$
	\end{itemize}
	% ---
	 Thus, merging \eqref{Eq.TypeII-3} and \eqref{Ineq.Radius_Event}, we can establish the following bound for $\mathbbmss{E}_1:$
	% ---
	\begin{align}
		\label{Ineq.Event_E1}
		\Pr(\mathbbmss{E}_1) \stackrel{(a)}{\leq} \Pr \Big( \sum_{t=1}^{\bar{n}} Z_{\rm w,t}^2 - 1 \leq - \bar{n}\delta_n \Big)
		\stackrel{(b)}{\leq} \frac{\sum_{t=1}^{\bar{n}} \text{Var} [Z_{\rm w,t}^2] }{\bar{n}^2\delta_n^2}
		\leq \frac{2}{\bar{n}\delta_n^2}\stackrel{(c)}{\leq} \frac{9}{8a^2C_{\sigma_{\rm max}}^2 n^{2\kappa + \mu + b}} \triangleq \zeta_1,
	\end{align}
	% ---
	where $(a)$ uses \eqref{Ineq.Radius_Event}, $(b)$ employs the Chebyshev's inequality and $(c)$ follows by similar arguments as provided in \eqref{Ineq.TypeI_Final}. Therefore, employing the upper bounds given in \eqref{Ineq.Event_E0_2}, \eqref{Eq.TypeIIError-E_0+E_1} and \eqref{Ineq.Event_E1} yields
	% ---
	\begin{align}
		P_{e,2}(i,j) \leq \Pr(\mathbbmss{E}_0) + \Pr(\mathbbmss{E}_1) \leq \zeta_0 + \zeta_1 \leq e_2,
	\end{align}
	% ---
	hence, $P_{e,2}(i,j) \leq e_2$ holds for sufficiently large $n$ and arbitrarily small $e_2 > 0.$ We have thus shown that for every $e_1,e_2 > 0$ and sufficiently large $n$, there exists an $(n, M(n,R),\allowbreak K(n,\kappa), \allowbreak e_1, e_2)$-DI code. This completes the achievability proof of Theorem~\ref{Th.Colored-Capacity}.

	\subsection{Upper Bound (Converse Proof)}
	\label{Subsec.Converse}
	For brevity in the derivations of Lemma~\ref{Lem.Converse} and to facilitate the subsequent analysis, we adopt the following notational conventions:
	% ---
	\begin{itemize}[leftmargin=*]
		\item $A_{\rm max} = KLP_{\rm max} = \mathcal{O}(n^{\kappa}).$
		\item $Y_t(i) = c_{i,t}^{\fh} + Z_t,\, \forall t \in [\![\bar{n}]\!]$ denote the channel output at time $t$ \emph{conditioned} that $\fx=\fc_i$ was sent.
		\item $c_{i,t}^{\fh} \triangleq \sum_{k=0}^{K-1} h_k c_{i,t-k} ,\, \forall t \in [\![\bar{n}]\!]$ is the convoluted symbol, i.e., the linear combination of $\fc_i$ and $\fh.$
		\item $(\fy - \fc_{k}^{\fh})_{\rm w} \triangleq \fSigma^{-1/2} ( \fy - \fc_{k}^{\fh} ), \, \forall k \in \{i,j\}.$
		\item $\fd_{i,j} \triangleq \fSigma^{-1/2} ( \fc_i^{\fh} - \fc_j^{\fh} ).$
		\item $\mathbbmss{C}_{\text{\tiny conv}} \triangleq \big\{ \fc_i \in \mathbb{R}^n:\; |c_{i,t}| \leq P_{\rm max} , \forall \, i \in [\![M]\!],\, \forall t \in [\![n]\!] \big\} .$
		\item $\mathbbmss{C}_{\text{\tiny conv}}^{\fh} \triangleq \big\{ \fc_i^{\fh} \in \mathbb{R}^{\bar{n}} :\, c_{i,t}^{\fh} \triangleq \sum_{k=0}^{K-1} h_k c_{i,t-k},\, \fc_i \in \mathbbmss{C}_{\text{\tiny conv}}, \forall \, i \in [\![M]\!],\, \forall t \in [\![\bar{n}]\!] \big\}.$
	\end{itemize}
	% ---
	\begin{customlemma}{5}
		\label{Lem.Converse}
		Suppose that $R$ is an achievable identification rate for $\bG.$ Let $\{(\mathbbmss{C}_{\text{\tiny conv}}^{(n)},\D^{(n)})\}_{n\in\mathbb{N}}$ be a sequence of $(n,\allowbreak M(n,\allowbreak R),\allowbreak K(n,\allowbreak \kappa),\allowbreak e_1^{(n)},\allowbreak e_2^{(n)})$-DI codes, where $K(n,\kappa)=n^{\kappa}$ for some $\kappa \in [0,1)$, and the error probabilities $e_1^{(n)}$ and $e_2^{(n)}$ both vanish as $n \to \infty.$ Then, for sufficiently large $n,$ the convoluted codebook $\mathbbmss{C}_{\text{\tiny conv}}^{\fh}$ satisfies the following property: any two distinct codewords $\fc_{i_1}^{\fh}$ and $\fc_{i_2}^{\fh}$ in $\mathbbmss{C}_{\text{\tiny conv}}^{\fh}$, with $i_1,i_2 \in [\![M]\!]$ and $i_1 \neq i_2$, are separated by a distance of at least
		% ---		
		\begin{align}
			\label{Ineq.Conv_Distance}
			\big\| \fc_{i_1}^{\fh} - \fc_{i_2}^{\fh} \big\| \geq \sqrt{\bar{n}\epsilon_n'} \triangleq \alpha_n,
		\end{align}
		% ---
		where $\epsilon_n' = a/\bar{n}^{2(1+ (\mu/2) + b)}$ with $b>0$ being an arbitrarily small constant.
	\end{customlemma}
	% ---
	\begin{proof}
		The proof is provided in Appendix~\ref{App.Converse_Proof}.
	\end{proof}
	% ---
	
	We next apply Lemma~\ref{Lem.Converse} to derive an upper bound on the identification capacity. Since the minimum distance of the convoluted codebook is $\alpha_n$, one can place non-overlapping spheres $\mathcal{S}_{\mathbf{c}_i^{\mathrm{h}}}(n,\alpha_n)$ centered at points in $\mathbbmss{C}_{\text{\tiny conv}}^{\mathrm{h}}$. These spheres are generally inscribed within the hypercube $\mathbbmss{Q}_{\mathbf{0}}(\bar{n},A_{\rm max} + 2 r_0)$. Following the reasoning in \cite{Salari26}, such a packing is typically not saturated; nevertheless, using the same approach, the number of codewords, $M,$ is bounded by
	% ---
	\begin{align}
		\label{Ineq.Codebook_Size_UB}
		M = \frac{\text{Vol}\left(\bigcup_{i=1}^{M} \S_{\fc_i}^{\fh}(\bar{n},r_0) \right)}{\text{Vol}(\S_{\fc_1}^{\fh}(\bar{n},r_0))} \stackrel{(a)}{\leq} \frac{\Updelta_n(\mathscr{S}) \cdot \text{Vol}\big( \mathbbmss{Q}_{\f0}(\bar{n},A_{\rm max} + 2r_0) \big)}{\text{Vol}(\S_{\fc_1}^{\fh}(\bar{n},r_0))}
		\stackrel{(c)}{\leq} 2^{-0.599\bar{n}} \cdot \frac{(A_{\rm max}+2r_0)^{\bar{n}}}{\text{Vol}(\S_{\fc_1}^{\fh}(\bar{n},r_0))},
	\end{align}
	% ---
	where $(a)$ holds since a saturated packing encompasses the maximum possible number of sphere, $(b)$ conforms the density definition and $(c)$ exploits  \eqref{Ineq.Density} and the following: $$\mathbbmss{C}_{\text{\tiny conv}}^{\fh} \subseteq \mathbbmss{Q}_{\f0}(\bar{n},A_{\rm max} + 2r_0) = \big\{ \fc_i^{\fh} \in \mathbb{R}^{\bar{n}} \hspace{-.4mm}:\hspace{-.2mm} - (A_{\rm max} + r_0) \leq c_{i,t}^{\fh} \leq A_{\rm max} + r_0, \, \forall \, i \in [\![M]\!], \, \forall \, t \in [\![\bar{n}]\!] \big\},$$ which implies $\text{Vol}( \mathbbmss{C}_{\text{\tiny conv}}^{\fh} ) \leq \text{Vol}( \mathbbmss{Q}_{\f0}(\bar{n},A_{\rm max} + 2r_0)) = (A_{\rm max} + 2r_0)^{\bar{n}}.$ Thereby,
	% ---
	\begin{align}
		\label{Ineq.Log_M_UB}
		\log M \leq \bar{n} \log (A_{\rm max} + 2r_0) - \bar{n} \log r_0 - \bar{n} \log \sqrt{\pi} + \frac{1}{2} \bar{n} \log \bar{n} + \mathcal{O}(\bar{n}).		 
	\end{align}
	% ---
	Now, for $r_0 = \sqrt{\bar{n}\epsilon_n'} = \sqrt{a}/\bar{n}^{\frac{1+\mu+2b}{2}}$ and $A_{\rm max} = KLP_{\rm max} = \mathcal{O}(n^{\kappa}),$ we obtain
	% ---
	\begin{align}
		\label{Ineq.Log_M_UB2}
		\log M \leq n \log n^{\kappa}LP_{\rm max} + K \log n^{\kappa}LP_{\rm max} + \Big( \frac{2+\mu+2b}{2} \Big) \bar{n} \log \bar{n} + \mathcal{O}(\bar{n}),
	\end{align}
	% --- \bar{n} \log \Big(1 + \frac{2a}{LP_{\rm max}n^{\kappa}\bar{n}^{\frac{1+\mu+2b}{2}}} \Big)
	where the dominant term scales as $\bar{n} \log \bar{n}$. Noting that $\bar{n} \log \bar{n} \sim n \log n$, we choose $M = 2^{(n \log n) R}$, resulting in
	% ---
	\begin{align}
		\label{Ineq.Rate_UB}
		R \leq \frac{1}{n \log n} \Big[ \Big( \frac{2\kappa + 2 + \mu + 2b}{2} \Big) \, n \log n + \kappa n \log P_{\rm max}  + o(n\log n) \Big] ,
	\end{align}
	% ---
	which tends to $1 + \kappa + (\mu/2) + b$ as $n \to \infty$ and $b \to 0.$ Now, since $b>0$ is arbitrarily small, an achievable rate must satisfy $R \leq 1 + \kappa + (\mu/2)$ This completes the proof of Theorem~\ref{Th.Colored-Capacity}.
	
	\section{Interpretation of the Capacity Bounds}
	\label{Sec.Interpretations}
	Theorem \ref{Th.Colored-Capacity} reveals that the colored Gaussian channel $\bG$ is parameterized by $(\kappa,\mu),$ that is, we have a mapping: $(\kappa,\mu) \mapsto \bG(\kappa,\mu).$ This observation implies that the channel is not a single fixed object; instead, it belongs to a family of channels indexed by the values of $\kappa$ and $\mu.$ Let define the admissible parameter region by 
	\[
	\mathcal{P}	= \left\{
	(\kappa,\mu)\in\mathbb{R}^2 :
	0 \le \kappa < \frac12,\;
	0 \le \mu < \frac12 ,\;
	0 \leq \kappa+\mu < \frac{1}{2}
	\right\}.
	\]
	Then, Theorem \ref{Th.Colored-Capacity} holds for the family of channels $\{ \bG(\kappa,\mu): (\kappa,\mu) \in \P \}$ and we can define a function:
	\[
	\mathbb{C}_{\rm I}(\bG;\kappa,\mu) := \mathbb{C}_{\rm I}(\bG(\kappa,\mu)).
	\]
	That is, since the channel $\bG$ is parameterized by $(\kappa,\mu)$, the identification capacity induces a scalar-valued function $\mathbb{C}_{\rm I}(\bG;\kappa,\mu).$ The bounds given in Theorem \ref{Th.Colored-Capacity} therefore describe pointwise constraints on this function over the parameter space $\P$ and establishe only a pointwise characterization of the identification capacity $\mathbb{C}_{\rm I}(\bG;\kappa,\mu).$ Specifically, it implies that
	\[
	\mathbb{C}_{\rm I}(\bG;\kappa,\mu) \in [L(\kappa,\mu),U(\kappa,\mu)],
	\]
	where the achievable and converse bounds are given by
	% ---
	\begin{align}
		L(\kappa,\mu) & = (1 - 2(\kappa + \mu))/ 4 \quad \text{ and } \quad U(\kappa,\mu) = 1+\kappa + \mu/2.
	\end{align}
	% ---	 
	Thus, Theorem \ref{Th.Colored-Capacity} does not determine the capacity function $\mathbb{C}_{\rm I}(\bG;\kappa,\mu)$ uniquely. Rather, it specifies a family of admissible functions consistent with the achievable and converse bounds. Defining the set of capacity functions
	\[
	\mathcal{F}_{\mathbb{C}_{\rm I}}
	=
	\left\{
	\mathbb{C}_{\rm I}:\mathcal{P}\to\mathbb{R}
	:\;
	L(\kappa,\mu)\le \mathbb{C}_{\rm I}(\bG;\kappa,\mu) \le U(\kappa,\mu),
	\ \forall (\kappa,\mu)\in\mathcal{P}
	\right\},
	\]
	the true (unknown) identification-capacity function satisfies
	\[
	\mathbb{C}_{\rm I}(\bG;\kappa,\mu) \in \mathcal{F}_{\mathbb{C}_{\rm I}}.
	\]
	Obesrve that both $L(\kappa,\mu)$ and $U(\kappa,\mu)$ are affine functions of $(\kappa,\mu)$. Moreover, $L(\kappa,\mu)$ is strictly decreasing in both $\kappa$ and $\mu$, whereas $U(\kappa,\mu)$ is strictly increasing in both parameters. The resulting bound gap is
	\[
	\Delta(\kappa,\mu) \triangleq U(\kappa,\mu) - L(\kappa,\mu) = \frac34 + \frac32\kappa + \mu
	\]
	which increases monotonically with $\kappa$ and $\mu$. It is important, however, not to attribute any structural property of the true capacity function $\mathbb{C}_{\rm I}(\bG)$ to the opposite monotonic behavior of the two bounding functions. The inequalities above impose only pointwise constraints on the admissible values of $\mathbb{C}_{\rm I}(\bG)$; they do not constrain/force how the capacity varies across neighboring parameter values. Consequently, the signs of the partial derivatives or directional derivatives of $\mathbb{C}_{\rm I}(\bG)$ are not determined by the bounds. In particular, the available bounds do not permit one to infer monotonicity, convexity, concavity, or any other regularity property (e.g., curvature, continuity, etc) of the identification capacity. The function $\mathbb{C}_{\rm I}(\bG)$ may increase, decrease, remain constant, or even exhibit non-monotonic behavior as a function of $\kappa$ and $\mu$, while remaining fully consistent with the available bounds. Therefore, the opposite monotonic behavior of the lower and upper bounds in $\kappa$ and $\mu$ should not be interpreted as evidence that the true identification capacity necessarily deteriorates or improves as either parameter grows. We note that even arbitrarily tight bounds $\Delta(\kappa,\mu) \to 0$ do not guarantee monotonicity unless the bounding functions themselves are order-preserving in the same direction.
	
	\textbf{Achievability:}  On the achievability side, increasing ISI and correlation make the effective noise geometry more anisotropic. In the absence of correlation, the covariance ellipsoid is closer to spherical, resulting in approximately uniform spherical noise that can be packed efficiently while preserving distinguishability. As the correlation becomes stronger, however, the covariance ellipsoid becomes increasingly elongated along certain directions. A similar effect can arise from ISI, which introduces additional structure and directional dependencies in the noise. These elongated noise neighborhoods increase the risk of overlap between codeword neighborhoods. To guarantee small false-positive and false-negative probabilities for all message pairs, larger separations between codewords are therefore required, reducing packing efficiency and resulting in a weaker lower bound. Consequently, the decrease in the achievability bound should be interpreted as a limitation of the particular coding and decoding construction employed in the proof rather than as evidence of a fundamental loss in capacity.
	
	\textbf{Converse:} On the converse side, the sphere-packing and volume-counting arguments identify a larger admissible geometric space in the presence of ISI and correlation, leading to a larger upper bound. However, without a corresponding achievability result, there is no evidence that coding schemes can actually exploit this enlarged space to reliably distinguish more messages. Thus, the increase should be viewed as a feature of the converse analysis, i.e., a volume-counting effect, rather than a demonstrated operational gain.
	
	\textbf{Interpretations:} The dependence of our upper and lower bounds on the ISI and correlation parameters should be interpreted with caution. The opposite behavior of the achievability and converse bounds is best viewed as an indication that the current analytical proof techniques are not yet sufficient to characterize the true asymptotic behavior of the identification capacity, rather than as evidence of an intrinsic benefit or penalty associated with ISI or correlation. Since the bounds do not match and no explicit capacity-achieving code construction is available, neither trend can be taken as evidence of the true behavior of the identification capacity. Clarifying the actual role of these channel characteristics will require improved achievability schemes, explicit capacity-achieving code construction, tighter converses or a characterization based on more intrinsic spectral properties of the underlying noise.
	
	\subsection{Parameterized Constructions with Prescribed Monotonicity}
	In the following, we present two parameterized capacity constructions that satisfy the derived bounds while exhibiting opposite monotonicity with respect to $\kappa$ and $\mu.$ The first is increasing in both parameters, whereas the second is decreasing. These artificial examples illustrate that the current bounds alone do not determine the qualitative dependence of the true identification capacity on ISI and correlation, since markedly different behaviors can be consistent with the same bounds. Let define the parameterized weight familiy as follows
	\[
	w_{a,b,c}(\kappa,\mu)
	=
	a+b\kappa+c\mu,
	\qquad a,b,c>0,
	\]
	and recall that
	\[
	L(\kappa,\mu)
	=
	\frac{1-2(\kappa+\mu)}{4},
	\qquad
	U(\kappa,\mu)
	=
	1+\kappa+\frac{\mu}{2},
	\]
	with $\Delta(\kappa,\mu)
	=
	U(\kappa,\mu)-L(\kappa,\mu)
	= 3/4 + 3\kappa/2 + \mu.$
	Since
	\[
	\frac{\partial L}{\partial\kappa}
	=
	\frac{\partial L}{\partial\mu}
	=
	-\frac12,
	\qquad
	\frac{\partial U}{\partial\kappa}
	=
	1,
	\qquad
	\frac{\partial U}{\partial\mu}
	=
	\frac12,  \qquad \frac{\partial \Delta}{\partial\kappa}
	=
	\frac32,
	\qquad
	\frac{\partial \Delta}{\partial\mu}
	=
	1,
	\]
	the monotonicity conditions can be derived explicitly.
	
	\textbf{Construction 1 - Increasing in both \texorpdfstring{$\kappa$}{kappa} and \texorpdfstring{$\mu$}{mu}}: Let define
	% ---
	\begin{align}
		\mathbb{C}_{\rm I}^{\uparrow}(\bG;\kappa,\mu)
		&=
		L(\kappa,\mu)
		+
		w_{a,b,c}(\kappa,\mu)\,
		\Delta(\kappa,\mu)
		\nonumber\\
		&=
		\frac{1-2(\kappa+\mu)}{4}
		+
		(a+b\kappa+c\mu)
		\left(
		\frac34+\frac32\kappa+\mu
		\right).
	\end{align}
	Differentiating yields
	\begin{align}
		\frac{\partial \mathbb{C}_{\rm I}^{\uparrow}(\bG;\kappa,\mu)}{\partial\kappa}
		&=
		-\frac12
		+
		b\,\Delta(\kappa,\mu)
		+
		\frac32\bigl(a+b\kappa+c\mu\bigr),
		\\[2mm]
		\frac{\partial \mathbb{C}_{\rm I}^{\uparrow}(\bG;\kappa,\mu)}{\partial\mu}
		&=
		-\frac12
		+
		c\,\Delta(\kappa,\mu)
		+
		\bigl(a+b\kappa+c\mu\bigr).
	\end{align}
	Since all remaining terms are nonnegative and increase with
	\(\kappa\) and \(\mu\), it suffices to check positivity at
	\((0,0)\). Using \(\Delta(0,0)=\frac34\), we obtain
	\[
	\frac{\partial \mathbb{C}_{\rm I}^{\uparrow}(\bG;0,0)}{\partial\kappa}
	=
	-\frac12+\frac34\,b+\frac32\,a \quad \text{ and } \quad \frac{\partial \mathbb{C}_{\rm I}^{\uparrow}(\bG;0,0)}{\partial\mu}
	=
	-\frac12+\frac34\,c+a.
	\]
	Therefore, sufficient conditions for monotonic increase are
	\[
	-\frac12+\frac34\,b+\frac32\,a>0,
	\qquad
	-\frac12+\frac34\,c+a>0.
	\]
	A simple admissible parameter region is $$\Big\{ a>\frac12, b>0, c>0. \Big\}$$
	Under these conditions,
	\[
	\frac{\partial \mathbb{C}_{\rm I}^{\uparrow}(\bG;\kappa,\mu)}{\partial\kappa}>0,
	\qquad
	\frac{\partial \mathbb{C}_{\rm I}^{\uparrow}(\bG;\kappa,\mu)}{\partial\mu}>0,
	\qquad
	\forall\,(\kappa,\mu)\in\mathcal P.
	\]
	
	\textbf{Construction 2 - Decreasing in both \texorpdfstring{$\kappa$}{kappa} and \texorpdfstring{$\mu$}{mu}}: Let define
	% ---
	\begin{align}
		\mathbb{C}_{\rm I}^{\downarrow}(\bG;\kappa,\mu)
		&=
		L(\kappa,\mu)
		+
		\bigl(1-w_{a,b,c}(\kappa,\mu)\bigr)
		\Delta(\kappa,\mu) =
		U(\kappa,\mu) - w_{a,b,c}(\kappa,\mu)\,
		\Delta(\kappa,\mu)
		\nonumber\\& 
		=
		1+\kappa+\frac{\mu}{2}
		-
		(a+b\kappa+c\mu)
		\left(
		\frac34+\frac32\kappa+\mu
		\right).
	\end{align}
	Differentiating yields
	\begin{align}
		\frac{\partial \mathbb{C}_{\rm I}^{\downarrow}(\bG;\kappa,\mu)}{\partial\kappa}
		&=
		1
		-
		b\,\Delta(\kappa,\mu)
		-
		\frac32\bigl(a+b\kappa+c\mu\bigr),
		\\[2mm]
		\frac{\partial \mathbb{C}_{\rm I}^{\downarrow}(\bG;\kappa,\mu)}{\partial\mu}
		&=
		\frac12
		-
		c\,\Delta(\kappa,\mu)
		-
		\bigl(a+b\kappa+c\mu\bigr).
	\end{align}
	Since the negative terms increase in magnitude with
	\(\kappa\) and \(\mu\), it suffices to verify negativity at
	\((0,0)\). Using \(\Delta(0,0)=\frac34\),
	\[
	\frac{\partial \mathbb{C}_{\rm I}^{\downarrow}(\bG;0,0)}{\partial\kappa}
	=
	1-\frac34\,b-\frac32\,a, \quad \text{ and } \quad \frac{\partial \mathbb{C}_{\rm I}^{\downarrow}(\bG;0,0)}{\partial\mu}
	=
	\frac12-\frac34\,c-a.
	\]
	Therefore, sufficient conditions for monotonic decrease are
	\[
	1-\frac34\,b-\frac32\,a<0,
	\qquad
	\frac12-\frac34\,c-a<0.
	\]
	A simple admissible parameter region is $$\Big\{ a>\frac23, b>0, c>0. \Big\}$$
	Under these conditions,
	\[
	\frac{\partial \mathbb{C}_{\rm I}^{\downarrow}(\bG;\kappa,\mu)}{\partial\kappa}<0,
	\qquad
	\frac{\partial \mathbb{C}_{\rm I}^{\downarrow}(\bG;\kappa,\mu)}{\partial\mu}<0,
	\qquad
	\forall\,(\kappa,\mu)\in\mathcal P.
	\]
	
In the following, we set the parameters of the weight function to $a = 0.72$, $b = 0.25$, and $c = 0.20$, and illustrate the corresponding capacity functions over the admissible region. To highlight the monotonic behaviour, we present surface plots for the proposed constructions in Figures \ref{Fig.Increasing}, \ref{Fig.Decreasing}, and \ref{Fig.Both}.

	\begin{figure}[H]
		\centering
		\includegraphics[scale=.62]{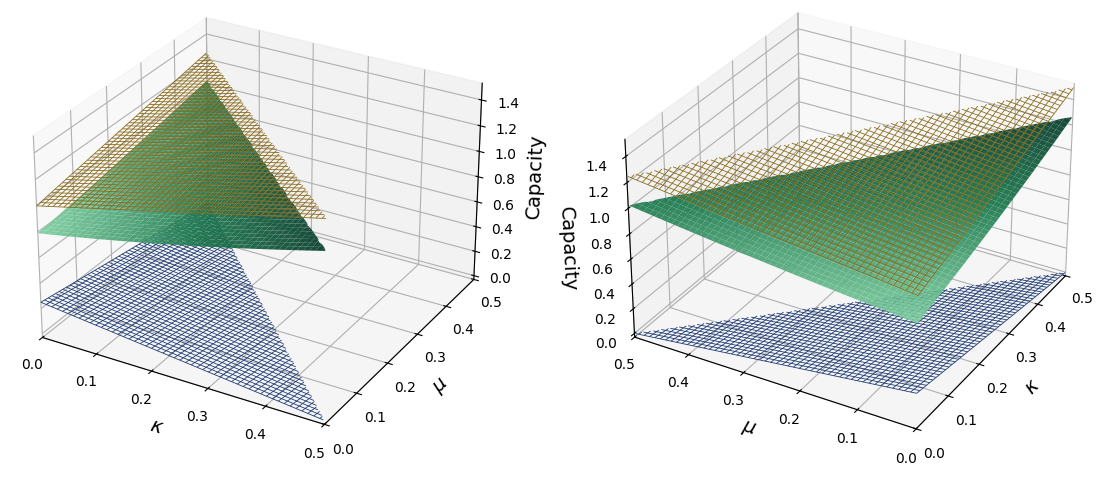}
		\vspace{-12mm}
		\caption{The increasing capacity surface shows that capacity rises monotonically as both parameters $\kappa$ and $\mu$ increase, driven by the weighting function $w = a + b\kappa + c\mu$, where $a=0.72$, $b=0.25$, and $c=0.20$. The capacity is defined as a convex combination of the lower and upper bounds through the form $\mathbb{C}_{\rm I}^{\uparrow}(\bG;\kappa,\mu)=L(\kappa,\mu)+w_{a,b,c}(\kappa,\mu) \Delta(\kappa,\mu).$ This formulation shows how $w$ modulates the transition between $L(\kappa,\mu)$ and $U(\kappa,\mu)$ over the admissible domain.}
		\label{Fig.Increasing}
	\end{figure}
	% ---
	% ---
	\begin{figure}[H]
		\centering
		\includegraphics[scale=0.62]{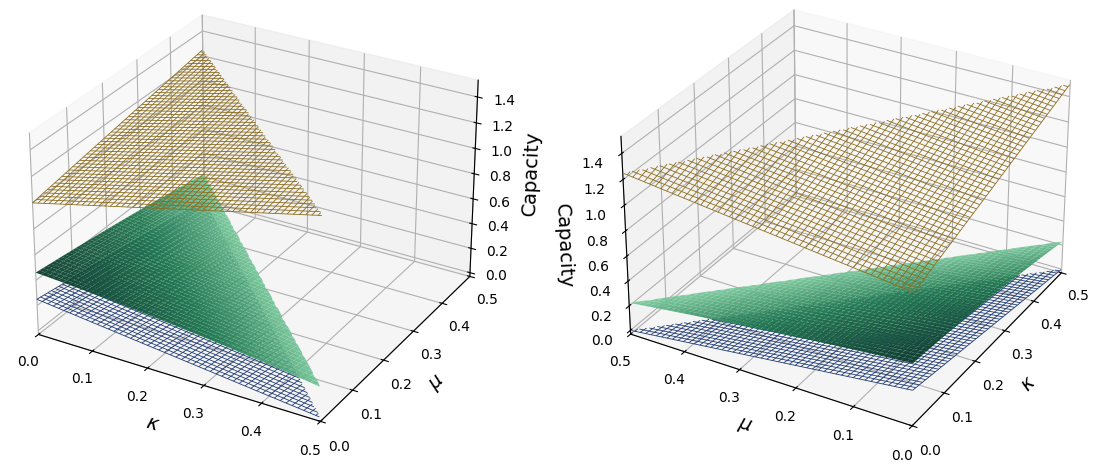}
		\vspace{-12mm}
		\caption{The decreasing capacity surface illustrates a monotonic decline in capacity as $\kappa$ and $\mu$ vary. The weighting function $w = a + b\kappa + c\mu$, where $a=0.72$, $b=0.25$, and $c=0.20$, governs the redistribution between the lower and upper bounds across the admissible domain. The capacity is defined as $\mathbb{C}_{\rm I}^{\downarrow}(\bG;\kappa,\mu)=L(\kappa,\mu)+ \big(1-w_{a,b,c}(\kappa,\mu)\bigr)\Delta(\kappa,\mu).$ This formulation captures how increasing influence of the weighting function alters the balance between $L(\kappa,\mu)$ and $U(\kappa,\mu)$, leading to a systematic reduction in the resulting capacity over the domain.}
		\label{Fig.Decreasing}
	\end{figure}
	% ---
		% ---
	\begin{figure}[H]
		\centering
		\includegraphics[scale=0.325]{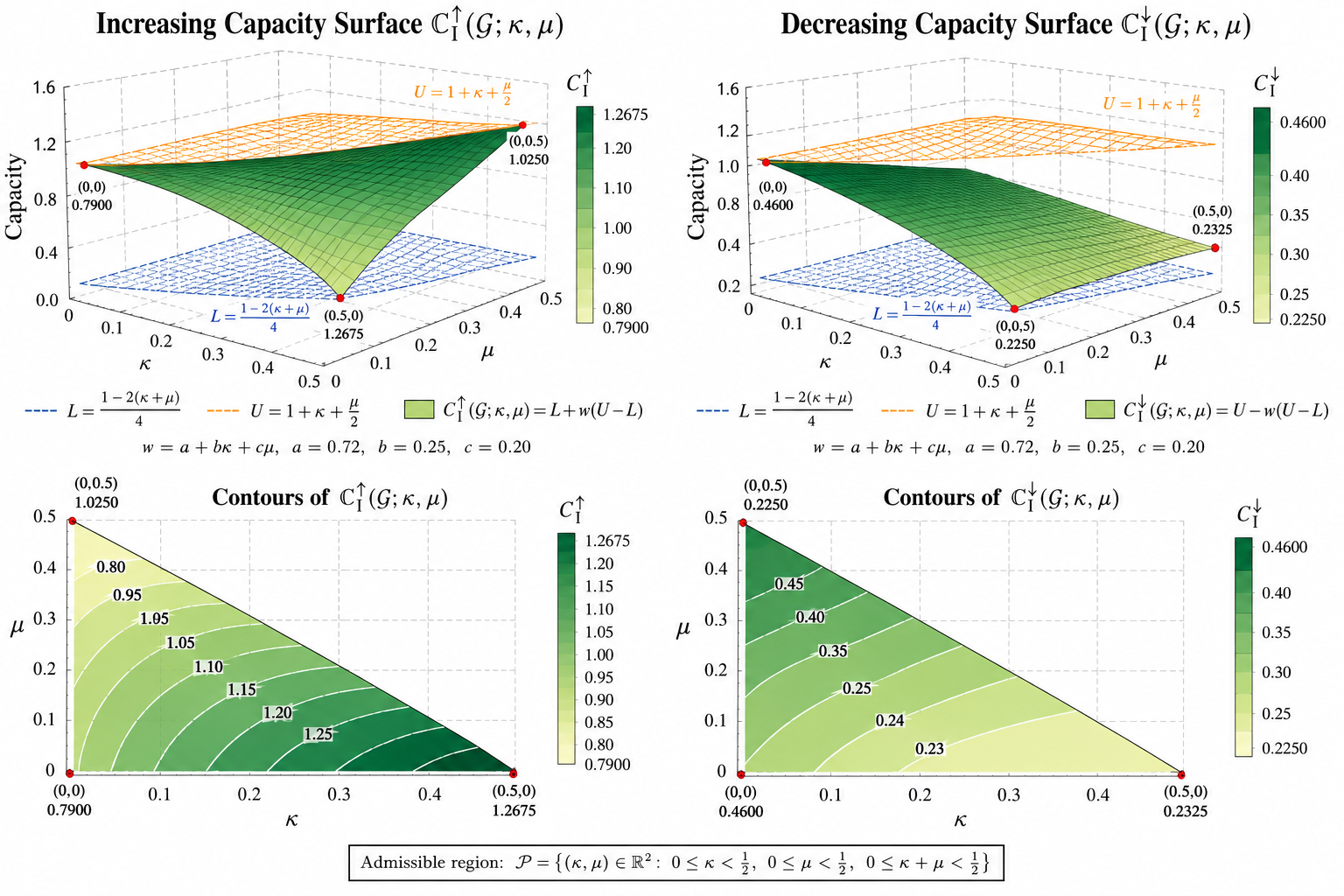}
		\caption{Increasing and decreasing capacity surfaces, together with their contour plots over the admissible region $\P$. The dashed blue and orange wireframes represent the lower and upper bounds $L=(1-2(\kappa+\mu))/4$ and $U=1+\kappa+\mu/2,$ respectively, while the labeled vertices indicate the exact capacity values at the boundary points of $\P.$}
		\label{Fig.Both}
	\end{figure}
	% ---

	% ---
	\section{Conclusions}
	\label{Sec.Conclusion}
	
	This work provides a rigorous treatment of the identification problem over the colored Gaussian channel with ISI, extending the classical memoryless \cite{Salariseddigh_ITW} and white noise model \cite{Salari26} to more realistic wireless settings. We show that reliable identification is achievable with super-exponential codebooks of size $M = 2^{(n \log n)R},$ even when the number of ISI taps grows sub-linearly in $n.$ In addition, we derive explicit lower and upper bounds on the identification rate $R$ as functions of the ISI growth rate $\kappa$ and the singular value growth rate $\mu.$ These results establish fundamental limits for identification over channels with both memory and colored noise, and point to extensions in channels with spectral nulls, multi-user scenarios, finite-blocklength analysis, slow or fast fading settings, exponentially bounded singular value spectrum regimes, and rank-deficient covariance matrices.

	\section{Acknowledgments}
	The authors would like to thank Professor Dr.-Ing. Dr. rer. nat. Holger Boche, Dr. Jonathan Huffmann and Dr. Johannes Rosenberger from Technical University of Munich for helpful discussions and insights concerning colored Gaussian channels.

	\appendices

	\section{Analysis of Whitening Noise Transformation}
	\label{App.WNA}
	
	In the following, we establish that the squared Mahalanobis distance $\fd_{j}^2$ in \eqref{Eq.d_mah} for stochastic $\fY$ follows a chi-squared distribution with $\bar{n}$ degree of freedom, i.e., $\fd_{j}^2 \sim \chi_{\bar{n}}^2.$

	\begin{proof}
		We start by decomposing $\fd_j^2 $ as follows
		\begin{align}
			\label{Eq.Eq.WT}
			\fd_j^2 = (\fY - \fx^{\fh})^T \fSigma^{-1} (\fY - \fx^{\fh}) \stackrel{(a)}{=} \big( \fSigma^{-1/2}(\fY - \fx^{\fh}) \big)^T \big( \fSigma^{-1/2}(\fY - \fx^{\fh}) \big)
			%\nonumber\\&
			\triangleq \fZ_{\rm w}^T\fZ_{\rm w} = \| \fZ_{\rm w} \|^2,
		\end{align}
		% ---
		where $(a)$ holds since $\fSigma$ is symmetric. Observe that $\fZ_{\rm w} \triangleq \fSigma^{-1/2} (\fY - \fx^{\fh})$ in \eqref{Eq.Eq.WT} is a whitening transformation that generates standard Gaussian vector which is proved in the following: First, linearity of the expectation gives $\mathbb{E}[\fZ_{\rm w}] = \fSigma^{-1/2} \mathbb{E}[ \fY - \fx^{\fh} ] = 0.$ Second, note that
		% ---
		\begin{align}
			\text{Cov}[\fSigma^{-1/2} (\fY - \fx^{\fh})] \stackrel{(a)}{=} \fSigma^{-1/2} \fSigma (\fSigma^{-1/2})^T = \fSigma^{-1/2} \fSigma \fSigma^{-1/2} = (\fSigma^{-1/2} \fSigma^{1/2}) \cdot ( \fSigma^{-1/2}\fSigma^{1/2}) = \fI,
		\end{align}
		% ---
		where $(a)$ holds since $\text{Cov}[\fA_{m \times n} \fX_{n \times n}] = \fA \text{Cov}[\fX] \fA^T;$ see \cite{Papoulis02}, and $\text{Cov}[\fY - \fx^{\fh}] = \fSigma.$ Thereby, since the expectation of whitened vector $\fZ_{\rm w}$ is zero and the covariance matrix of $\fZ_{\rm w}$ is the identity matrix, we infer that $\fZ_{\rm w}$ is a standard Gaussian vector, i.e., $Z_{\rm w,t} \overset{\text{\tiny i.i.d}}{\sim} \N(0,1).$ Now, since $\fd_j^2 = \| \fZ_{\rm w} \| ^2 = \sum_{t=1}^{\bar{n}} Z_{\rm w,t}^2$ we conclude that $\fd_j^2 \sim \chi_{\bar{n}}^2.$
	\end{proof}
	% ---

	\section{Proof of Lemma~\ref{Lem.RQ}}
	\label{App.RQ}
	
	In the following, we use the spectral decomposition theorem \cite{Strang12} (see Appendix \ref{App.Spec_Decom_Th} for details) and standard techniques to show that the Rayleigh quotient is bounded between the smallest and largest eigenvalues of the matrix.
	% ---
	\begin{proof}
		Let $\mathbf A \in \mathbb R^{n \times n}$ be a symmetric matrix, and let its eigenvalues be ordered as
		$$
		\lambda_{\min} \triangleq \lambda_1 \le \lambda_2 \le \cdots \le \lambda_n \triangleq \lambda_{\max}.
		$$
		By the spectral theorem, there exists an orthogonal matrix $\mathbf Q$ (i.e., $\fQ^T \fQ = I$) and a diagonal matrix
		$
		\Lambda = \operatorname{diag}(\lambda_1,\dots,\lambda_n)
		$
		such that
		$
		\mathbf A = \mathbf Q \Lambda \mathbf Q^T.
		$ For any nonzero vector $\mathbf x \in \mathbb R^n$, define
		$
		\mathbf y \triangleq \mathbf Q^T \mathbf x.
		$
		Since $\mathbf Q$ is orthogonal,
		$
		\mathbf y^T \mathbf y = \mathbf x^T \mathbf x > 0,
		$
		and
		$$
		\mathbf x^T \mathbf A \mathbf x
		= \mathbf x^T \mathbf Q \Lambda \mathbf Q^T \mathbf x
		= \mathbf y^T \Lambda \mathbf y.
		$$
		Therefore, the Rayleigh quotient is given by
		% ---
		\begin{align}
			R(\mathbf x)
			\triangleq \frac{\mathbf x^T \mathbf A \mathbf x}{\mathbf x^T \mathbf x}
			= \frac{\mathbf y^T \Lambda \mathbf y}{\mathbf y^T \mathbf y}
			= \frac{\sum_{t=1}^n \lambda_t y_t^2}{\sum_{t=1}^n y_t^2}.
		\end{align}
		% ---
		Because $\lambda_{\min} \le \lambda_t \le \lambda_{\max}, \forall t \in [\![n]\!],$ multiplying by $y_t^2 \ge 0$ gives $\lambda_{\min} y_t^2 \le \lambda_t y_t^2 \le \lambda_{\max} y_t^2.$
		Summing over all $t \in [\![n]\!]$ yields
		% ---
		\begin{align}
			\lambda_{\min} \sum_{t=1}^n y_t^2
			\le
			\sum_{t=1}^n \lambda_t y_t^2
			\le
			\lambda_{\max} \sum_{t=1}^n y_t^2.
		\end{align}
		% ---
		Since $\sum_{t=1}^n y_t^2 > 0$, dividing through gives
		$$
		\lambda_{\min}
		\le
		\frac{\sum_{t=1}^n \lambda_t y_t^2}{\sum_{t=1}^n y_t^2}
		\le
		\lambda_{\max}.
		$$
		Hence, $\forall \mathbf x \neq \mathbf 0$ we have
		$$
		\lambda_{\min}
		\le
		\frac{\mathbf x^T \mathbf A \mathbf x}{\mathbf x^T \mathbf x}
		\le
		\lambda_{\max} .
		$$
		Thus, the Rayleigh quotient of a symmetric matrix is always bounded between its smallest and largest eigenvalues. Moreover,
		% ---
		\begin{itemize}
			\item $R(\mathbf x) = \lambda_{\max}$ if and only if $\mathbf x$ lies in the eigenspace corresponding to $\lambda_{\max}$,
			\item $R(\mathbf x) = \lambda_{\min}$ if and only if $\mathbf x$ lies in the eigenspace corresponding to $\lambda_{\min}$.
		\end{itemize}
		% ---
	\end{proof}
	% ---
	
	\section{Spectral Decomposition Theorem}
	\label{App.Spec_Decom_Th}
	
	In the following, we introduce the spectral decomposition theorem \cite{Strang12}.
	% ---
	\begin{theorem}
		Let $\mathbf A \in \mathbb R^{n \times n}$ be a real symmetric matrix. Then $\mathbf A$ has $n$ linearly independent eigenvectors $\mathbf q_1,\dots,\mathbf q_n \in \mathbb R^n$ that can be chosen to be orthonormal. Let $\mathbf Q \triangleq [\mathbf q_1 \ \mathbf q_2 \ \ldots \ \mathbf q_n]$
		be the matrix whose columns are these eigenvectors, and let $\Lambda \triangleq \operatorname{diag}(\lambda_1,\dots,\lambda_n),$
		where $\lambda_t$ is the eigenvalue corresponding to $\mathbf q_t$. Then $\mathbf Q$ is orthogonal, i.e., $\mathbf Q^{-1} = \mathbf Q^T,$
		and $\mathbf A$ admits the factorization
		% ---
		\begin{align}
			\mathbf A = \mathbf Q \Lambda \mathbf Q^T.
		\end{align}
		% ---
		Moreover, $\mathbf A$ admits the equivalent spectral representation
		% ---
		\begin{align}
			\mathbf A = \sum_{t=1}^n \lambda_t \mathbf q_t \mathbf q_t^T.
		\end{align}
		% --- 
	\end{theorem}
	% ---
	
	% ---
	\begin{proof}
		Since $\mathbf A$ is symmetric, it has $n$ linearly independent eigenvectors $\mathbf q_1, \mathbf q_2, \dots, \mathbf q_n$
		with corresponding eigenvalues $\lambda_1, \lambda_2, \dots, \lambda_n.$ This means that for each $t$, we have
		$$\mathbf A \mathbf q_t = \lambda_t \mathbf q_t.$$
		Next, we build a matrix from eigenvectors $\mathbf Q \triangleq [\mathbf q_1 \ \mathbf q_2 \ \cdots \ \mathbf q_n].$ We also build a diagonal matrix of eigenvalues $\Lambda \triangleq \text{diag}(\lambda_1, \lambda_2, \dots, \lambda_n).$ Observe that, the equation $\mathbf A \mathbf q_t = \lambda_t \mathbf q_t$ for all $t$ can be written at once as:
		\[
		\mathbf A \mathbf Q = \mathbf Q \Lambda.
		\]
		This works because multiplying $\mathbf A$ by each column of $\mathbf Q$ produces the same effect as multiplying by $\Lambda$. Since the eigenvectors are linearly independent, $\mathbf Q$ is invertible. So we multiply both sides on the right by $\mathbf Q^{-1}$ and obtain $\mathbf A \mathbf Q \mathbf Q^{-1} = \mathbf Q \Lambda \mathbf Q^{-1}.$		
		This simplifies to:
		\[
		\mathbf A = \mathbf Q \Lambda \mathbf Q^{-1}.
		\]
		Because $\mathbf A$ is symmetric, we can choose the eigenvectors to be orthonormal. This makes $\mathbf Q$ an orthogonal matrix, therefore $\mathbf Q^{-1} = \mathbf Q^T.$ Therefore,
		\[
		\mathbf A = \mathbf Q \Lambda \mathbf Q^T.
		\]
		\textbf{Interpretation:} A real symmetric matrix acts as a rotation into an eigenbasis, followed by scaling along orthogonal directions, followed by a rotation back. That is, $\mathbf A$ acts as rotation into eigenvector directions ($\mathbf Q^T$), scaling by eigenvalues ($\Lambda$), and rotating back by ($\mathbf Q$).
	\end{proof}

\section{Proof of Lemma~\ref{Lem.DB_Sing_Val}}
\label{App.B_Sing_Inv}

In the following, we employ the singular value decomposition (SVD) to derive the singular values of the inverse matrix depending to the original singular values.

\begin{proof}
	
	Let the SVD of $\fSigma$ be given by $\fSigma = \fU \flambda \fV^T$ where $\fU, \fV \in \mathbb{R}^{\bar{n} \times \bar{n}}$ are unitary matrices and $\flambda$ is a diagonal matrix consisting of all singular values, i.e., $\flambda = \text{diag}(\sigma_1(\fSigma), \dots, \sigma_{\bar{n}}(\fSigma)),$
	with $\sigma_1(\fSigma) \ge \cdots \ge \sigma_{\bar{n}}(\fSigma) > 0$. Next, we invert the decomposition. Since $\fSigma$ is invertible, all singular values are positive, thus,
	$\fSigma^{-1} = \text{diag}\!\left( \sigma_1^{-1}(\fSigma), \dots, \sigma_{\bar{n}}^{-1}(\fSigma) \right).$
	Using $(\fA\fB\fC)^{-1} = \fC^{-1}\fB^{-1}\fA^{-1}$ for matrices $\fA,\fB,\fC,$ we obtain $$\fA^{-1} = \fV \flambda^{-1} \fU^T,$$ which is a SVD of $\fSigma^{-1}$ since $\fV$ and $\fU^T$ are unitary and $\fSigma^{-1}$ is diagonal with positive entries. Therefore, $\fA^{-1} = \fV \fSigma^{-1} \fU^T$ is a SVD of $\fA^{-1}$, up to ordering. Now, we determine the singular values. Observe that the diagonal entries of $\fSigma^{-1}$ are $\sigma_1^{-1}(\fSigma), \dots, \sigma_{\bar{n}}^{-1}(\fSigma).$ Since $\sigma_1(\fSigma) \ge \cdots \ge \sigma_{\bar{n}}(\fSigma)$, we obtain $$\sigma_1^{-1}(\fSigma)\le \cdots \le \sigma_{\bar{n}}^{-1}(\fSigma).$$ Thus they are in increasing order. Next, arranging in decreasing order gives $$\sigma_1(\fSigma^{-1}) = \sigma_{\bar{n}}^{-1}(\fSigma) \ , \ \sigma_2(\fSigma^{-1}) = \sigma_{\bar{n}-1}^{-1}(\fSigma) \ , \, \dots \ , \ \sigma_{\bar{n}}(\fSigma^{-1}) = \sigma_1^{-1}(\fSigma).$$
	Thus, $\sigma_t(\fSigma^{-1}) = \sigma_{\bar{n}-t+1}^{-1}(\fSigma)\,, \forall t \in [\![\bar{n}]\!].$ Now, since $\sigma_1(\fSigma) \ge \cdots \ge \sigma_{\bar{n}}(\fSigma),$ the following bounds on the singular values of the inverse covariance matrix are obtained
	$$\sigma_1^{-1}(\fSigma) \le \sigma_t(\fSigma^{-1}) \le \sigma_{\bar{n}}^{-1}(\fSigma).$$
\end{proof}
% ---

	\section{Proof of Lemma~\ref{Lem.Converse}}
	\label{App.Converse_Proof}
	
	We establish Lemma~\ref{Lem.Converse} via a proof by contradiction. To this end, suppose that the condition in \eqref{Ineq.Conv_Distance} is violated, and show that this assumption leads to a contradiction. In particular, we prove that the sum of the type I and type II error probabilities converge to one, i.e., $\lim_{n \to \infty} \big[ P_{e,1}(i_1) + P_{e,2}(i_2,i_1) \big] = 1.$
	
	% ---
	\begin{proof}
	Fix $e_1$ and $e_2$. Let $\tau,\theta,\zeta>0$ be arbitrarily small. Assume to the contrary that 
	there exist two messages $i_1$ and $i_2$, where $i_1\neq i_2$, such that
	% ---
	\begin{align}
		\label{Eq.Alpha_nFast}
		\big\| \fc_{i_1}^{\fh} - \fc_{i_2}^{\fh} \big\| < \sqrt{\bar{n}\epsilon_n'} \triangleq \alpha_n = \sqrt{a}/\bar{n}^{\frac{1+ \mu + 2b}{2}} .
	\end{align}
	% ---
	Now let us define two subsets as follows
	% ---
	\begin{align}
		\label{Eq.Event_BC}
		\mathbbmss{D}_{i_1,i_2} & \triangleq \Big\{ \fy \in \mathbbmss{D}_{i_1}: \| (\fy - \fc_{i_2}^{\fh})_{\rm w} \| \leq \sqrt{\bar{n}( 1 + \zeta)} \Big\},
		\nonumber\\
		\mathbbmss{E}_{i_2} & \triangleq \Big\{ \fy \in \mathbb{R}^{\bar{n}}:
		\| (\fy - \fc_{i_1}^{\fh})_{\rm w} \| \leq \sqrt{\bar{n}(1 + \zeta)} \Big\} .
	\end{align}
	% ---
	Next, we can bound the type I error probability according to the events designed in \eqref{Eq.Event_BC} as follows
	% ---
	\begin{align}
		1-P_{e,1}(i_1) = \int_{\mathbbmss{D}_{i_1}} \hspace{-2mm} f_{\fZ}(\fy - \fc_{i_1}^{\fh}) d\fy & = \int_{\mathbbmss{D}_{i_1,i_2}} f_{\fZ}(\fy - \fc_{i_1}^{\fh}) d\fy + \int_{\mathbbmss{D}_{i_1} \setminus \mathbbmss{D}_{i_1,i_2}} f_{\fZ}(\fy - \fc_{i_1}^{\fh}) d\fy 
		\nonumber\\&
		\leq \int_{\mathbbmss{D}_{i_1,i_2}} f_{\fZ}(\fy - \fc_{i_1}^{\fh}) d\fy + \int_{\mathbbmss{E}_{i_2}^c} f_{\fZ}(\fy - \fc_{i_1}^{\fh}) d\fy ,
		\label{Eq.Pe1boundConv0Fast}
	\end{align}
	% ---
	where the last inequality holds since $\mathbbmss{D}_{i_1} \setminus \mathbbmss{D}_{i_1,i_2} \subset \mathbbmss{E}_{i_2}^c.$
	% ---
	Consider the second integral, for which the domain is $\mathbbmss{E}_{i_2}^c$. Then, by the triangle inequality
	% ---
	\begin{align}
		\| (\fy - \fc_{i_1}^{\fh})_{\rm w} \| & \geq \| (\fy - \fc_{i_2}^{\fh})_{\rm w} \| - \| \fd_{i,j} \|
		\nonumber\\&
		> \sqrt{\bar{n}( 1 + \zeta)} - \sigma_{\rm max}(\fSigma^{-1/2}) \| \fc_{i_1}^{\fh} - \fc_{i_2}^{\fh} \| \geq \sqrt{\bar{n}(1 + \zeta)} - \sigma_{\rm max}(\fSigma^{-1/2}) \alpha_n .
	\end{align}
	% ---
	The above inequality for $\eta<\frac{\zeta}{2}$ and sufficiently large $n,$ implies the following subset
	% ---
	\begin{align}
		\mathbbmss{F}_{i_1,i_2}^c = \Big\{\fy \in \mathbb{R}^{\bar{n}} \; : \, \| ( \fy - \fc_{i_1}^{\fh} )_{\rm w} \| > \sqrt{\bar{n}( 1 + \eta)} \Big\},
		\label{Eq.Regiong0}
	\end{align}
	% ---
	That is,
	% ---
	\begin{align}
		\Big\{\fy \in \mathbb{R}^{\bar{n}} \; : \, \| (\fy - \fc_{i_2}^{\fh})_{\rm w} ) \|  \geq
		\sqrt{\bar{n}( 1 + \zeta )} \Big\}   \overset{\text{implies}}{\longrightarrow}   \Big\{\fy \in \mathbb{R}^{\bar{n}} \; : \, \| (\fy - \fc_{i_1}^{\fh})_{\rm w} ) \|  \geq
		\sqrt{\bar{n}( 1 + \eta )} \Big\} .
	\end{align}
	% ---
	Thereby, we conclude that $\mathbbmss{F}_{i_1,i_2}^c \supset \mathbbmss{E}_{i_2}^c.$ Hence, the second integral in \eqref{Eq.Pe1boundConv0Fast} is bounded by
	% ---
	\begin{align}
		\label{Ineq.F_i_1_i_2_compl}
		\int_{\mathbbmss{F}_{i_1,i_2}^c} \hspace{-6mm} f_{\fZ}( \fy - \fc_{i_1}^{\fh}) d\fy = \Pr \Big( \| \fSigma^{-1/2} ( \fY(i_1) - \fc_{i_1}^{\fh} ) \| \hspace{-.7mm} > \hspace{-.7mm} \sqrt{\bar{n}(1 + \eta)} \Big)
		\hspace{-.7mm} \stackrel{(a)}{=} \hspace{-.7mm} \Pr \big( \bar{n}^{-1} \big\| \fZ_{\rm w} \big\|^2 - 1 > \eta \big) \stackrel{(b)}{\leq} \frac{2}{n\eta^2} \leq \tau,
	\end{align}
	% ---
	for sufficiently large $n,$ where $(a)$ follows by the substitution of $\fZ_{\rm w} \equiv \fSigma^{-1/2} ( \fY(i_1) - \fc_{i_1}^{\fh} )$ and $(b)$ holds by the Chebyshev's inequality and exploiting $n \leq \bar{n}$ and the following:
	% ---
	\begin{align} 
		\text{Var}[\bar{n}^{-1} \| \fZ_{\rm w} \|^2 - 1] = \bar{n}^{-2} \text{Var}[\| \fZ_{\rm w} \|^2 ] \stackrel{(a)}{=} \bar{n}^{-2} \sum_{t=1}^{\bar{n}} \text{Var} [Z_{\rm w,t}^2] \stackrel{(b)}{=} \bar{n}^{-2} \big( \sum_{t=1}^{\bar{n}} 3\sigma_{Z_{\rm w,t}}^4 - \sigma_{Z_{\rm w,t}}^2 \big) = 2\bar{n}^{-1},
	\end{align}
	% ---
	where $(a)$ invokes $Z_{\rm w,t} \overset{\text{\tiny i.i.d}}{\sim} \N(0,1)$ and $(b)$ holds since $\text{Var}[Z_{\rm w,t}^2] = \mathbb{E}[Z_{\rm w,t}^4] - (\mathbb{E}[Z_{\rm w,t}^2])^2 $ and $\mathbb{E}[Z_t^4] = 3\sigma_{Z_t}^4$ for $Z_t \overset{\text{\tiny i.i.d}}{\sim} \N(0,\sigma_{Z_t}^2)$ with setting $Z_t = Z_{\rm w,t}.$ Thus, merging \eqref{Eq.Pe1boundConv0Fast} and \eqref{Ineq.F_i_1_i_2_compl} and  gives
	%%%
	\begin{align}
		\label{Eq.ComplTypeIFast}
		1 - \tau - P_{e,1}(i_1) \leq \int_{\mathbbmss{D}_{i_1,i_2}} f_{\fZ}(\fy - \fc_{i_1}^{\fh}) d\fy.
	\end{align}
	%%%
	
	Now, we can focus on the inner integral with domain of $\mathbbmss{D}_{i_1,i_2}$, i.e., when
	% ---
	\begin{align}
		\| (\fy - \fc_{i_2}^{\fh})_{\rm w} \| \leq \sqrt{\bar{n}( 1 + \zeta)}.
		\label{Eq.ui2DistFast}
	\end{align}
	% ---
	Observe that, the absolute value of difference between noise distribution for distinct codewords reads
	% ---
	\begin{align}
		\label{Ineq.Error_Diff}
		\hspace{-2mm} \big| f_{\fZ}(\fy - \fc_{i_1}^{\fh}) - f_{\fZ}(\fy - \fc_{i_2}^{\fh}) \big|	= f_{\fZ}(\fy - \fc_{i_1}^{\fh}) \cdot \Big| 1 - \exp \big( - \big( \| (\fy - \fc_{i_2}^{\fh})_{\rm w} \|^2 - \| (\fy - \fc_{i_1}^{\fh})_{\rm w} \|^2 \big) / 2 \big) \Big| .
	\end{align}
	% ---
	Now, by the triangle inequality, we have $\| (\fy - \fc_{i_1}^{\fh})_{\rm w} \| \leq \| (\fy - \fc_{i_2}^{\fh})_{\rm w} \| + \|  \fd_{i,j} \|.$ Then, taking the square of both sides, we obtain
	% ---
	\begin{align}
		\label{Ineq.Squared_Tri_Ineq}
		\| (\fy - \fc_{i_1}^{\fh})_{\rm w} \|^2 &
		 \leq \| (\fy - \fc_{i_2}^{\fh})_{\rm w} \|^2 \hspace{-.5mm} + \hspace{-.5mm} \| \fd_{i,j} \|^2 + 2 \| (\fy - \fc_{i_2}^{\fh})_{\rm w} \| \cdot \| \fd_{i,j} \|
		\nonumber\\&
		\stackrel{(a)}{\leq} \| (\fy - \fc_{i_2}^{\fh})_{\rm w} \|^2 + \frac{a\sigma_{\rm max}^2(\fSigma^{-1/2})}{\bar{n}^{1+ \mu + 2b}} + \frac{2\sigma_{\rm max}(\fSigma^{-1/2}) \sqrt{a( 1 + \zeta)}}{\bar{n}^{\frac{\mu}{2} + b}} ,
	\end{align}
	% ---
	where $(a)$ holds by $\| \fd_{i,j} \| \leq \sigma_{\rm max}(\fSigma^{-1/2}) \| \fc_{i_1}^{\fh} - \fc_{i_2}^{\fh} \|,$ see \cite[Lem. 5]{Salariseddigh_affine_23_arXiv}, \eqref{Eq.Alpha_nFast}, \eqref{Eq.ui2DistFast} and exploiting $\alpha_n = \sqrt{a}/\bar{n}^{\frac{1+ \mu + 2b}{2}}.$ Next, to evaluate the behaviour of terms in \eqref{Ineq.Squared_Tri_Ineq} we use a helpful lemma which establish bounds on the singular values of the inverse square root of covariance matrix $\fSigma^{-1/2}.$
	
	% ---
	\begin{customlemma}{6}
		\label{Lem.Spec_Bound_Mat_Pow}
		Let $\fSigma$ be a symmetric and positive definite covariance matrix and assume that $\sigma_{\rm min}(\fSigma) \in \Omega(\bar{n}^{-\mu})$ and $\sigma_{\rm max}(\fSigma) \in \mathcal{O}(\bar{n}^{\mu/2}),$ then for any $p \in \mathbb{R},$ with constants $C_{\sigma_{\rm min}} > 0$ and $C_{\sigma_{\rm max}} > 0,$ respectively. Then, we have 
		% ---C_{\sigma_{\rm min}}^p \bar{n}^{-p\mu} \le \sigma_t(\fSigma^p) \le C_{\sigma_{\rm max}}^p \bar{n}^{p\mu/2}
		% ---
		\begin{equation}
			\begin{cases}
				C_{\sigma_{\rm min}}^p \bar{n}^{-p\mu} \le \sigma_t(\fSigma^p) \le C_{\sigma_{\rm max}}^p \bar{n}^{p\mu/2}, & p>0 ,
				%\sigma_{\rm min}(\fSigma^p) \in \Omega(\bar{n}^{-p\mu}) \quad \text{and} \quad \sigma_{\rm max}(\fSigma^p) \in \mathcal{O}(\bar{n}^{p\mu/2}).
				\\
				C_{\sigma_{\rm max}}^{-|p|} \bar{n}^{-|p|\mu/2} \le \sigma_t(\fSigma^p) \le C_{\sigma_{\rm min}}^{-|p|} \bar{n}^{|p|\mu}, & p<0 .
			\end{cases}
		\end{equation}
		% ---
	\end{customlemma}
	% ---
	% ---
	\begin{proof}
		The proof is provided in Appendix \ref{App.Spec_Bound_Mat_Pow}.
	\end{proof}
	% ---
	Next, employing Lemma \ref{Lem.Spec_Bound_Mat_Pow} with $p=1/2$ gives $\sigma_{\rm max}(\fSigma^{-1/2}) \le C_{\sigma_{\rm min}}^{-1/2} \bar{n}^{\mu/2}.$ Therefore, recalling \eqref{Ineq.Squared_Tri_Ineq} for sufficiently large $n,$ we obtain
	% ---
	\begin{align}
		\label{Ineq.Norm_S_Conv}
		\| ( \fy - \fc_{i_2}^{\fh} )_{\rm w} \|^2 - \| ( \fy - \fc_{i_1}^{\fh} )_{\rm w} \|^2 \leq \frac{a^2 C_{\sigma_{\rm min}}^{-1} \bar{n}^{\mu} }{\bar{n}^{1+\mu+2b}} + \frac{2C_{\sigma_{\rm min}}^{-1/2} \bar{n}^{\mu/2} \sqrt{a( 1 + \zeta)}}{\bar{n}^{\frac{\mu}{2} + b}}
		\le \theta .
	\end{align}
	% ---
	Hence, recalling \eqref{Ineq.Error_Diff} and \eqref{Ineq.Norm_S_Conv} yields
	%%%
	\begin{align}
		\label{Ineq.GaussianContinuityFast}
		\big| f_{\fZ}(\fy - \fc_{i_1}^{\fh}) - f_{\fZ}(\fy - \fc_{i_2}^{\fh}) \big| \leq  f_{\fZ}(\fy - \fc_{i_1}^{\fh}) \cdot \big| 1 - e^{\frac{\theta}{2\sigma_Z^2}} \big| \leq \tau f_{\fZ}(\fy - \fc_{i_1}^{\fh}),
	\end{align}
	%%%
	for sufficiently small $\theta>0$ such that $|1-e^{\frac{\theta}{2\sigma_Z^2}}| \leq \tau.$ Now, using \eqref{Eq.ComplTypeIFast} we have the following lower bound on the sum of the type I and type II error probabilities
	%%%
	\begin{align}
		P_{e,1}(i_1) + P_{e,2}(i_2,i_1) &\geq 1-\tau - \int_{\mathbbmss{D}_{i_1,i_2}} f_{\fZ}(\fy - \fc_{i_1}^{\fh})\,d\fy + \int_{\mathbbmss{D}_{i_1}} f_{\fZ}(\fy - \fc_{i_2}^{\fh}) \,d\fy
		\nonumber\\&
		\geq 1- \tau - \int_{\mathbbmss{D}_{i_1,i_2}} \big| (f_{\fZ}(\fy - \fc_{i_1}^{\fh}) - f_{\fZ}(\fy - \fc_{i_2}^{\fh})) \big| \,d\fy.
	\end{align}
	%%%
	Hence, by (\ref{Ineq.GaussianContinuityFast}),
	%%%
	\begin{align}
		P_{e,1}(i_1) + P_{e,2}(i_2,i_1) \geq 1- \tau -\tau \int_{\mathbbmss{D}_{i_1,i_2}} f_{\fZ}(\fy - \fc_{i_1}^{\fh}) d\fy \geq 1-2\tau,
	\end{align}
	%%%
	which leads to a contradiction for sufficiently small $\tau$ such that $2\tau > 1 - e_1 - e_2.$ Clearly, this is a contradiction since the error probabilities tend to zero as $n \rightarrow \infty.$ Thus, the assumption in \eqref{Eq.Alpha_nFast} is false. This completes the proof of Lemma~\ref{Lem.Converse}.
	
	\end{proof}
	% ---

	\section{Spectrum of Covariance Matrix Power}
	\label{App.Spec_Bound_Mat_Pow}
	In the following, we provide bounds on the singular values of whitening transform which is a fractional matrix power. Our proof method employs the spectral decomposition \cite[Ch. 9]{Kunze71} of a matrix which reduces the problem to scalar asymptotics, and thus matrix powers simply raise eigenvalues to the same power, preserving the order and the asymptotic structure. 
	% ---
	\begin{proof}
	Observe that if $\sigma_{\rm min}(\fSigma) \in \Omega(\bar{n}^{-\mu})$ and $\sigma_{\rm max}(\fSigma) \in \mathcal{O}(\bar{n}^{\mu/2})$ with constants $C_{\sigma_{\rm min}} > 0$ and $C_{\sigma_{\rm max}} > 0,$ respectively, so that for sufficiently large $\bar{n},$ for every $t \in [\![\bar{n}]\!]:$
	% ---
	\begin{align}
		\label{Ineq.DB_Sing}
		C_{\sigma_{\rm min}} \bar{n}^{-\mu} \le \sigma_t(\fSigma) \le C_{\sigma_{\rm max}} \bar{n}^{\mu/2}.
	\end{align}
	% ---
	Then, via spectral decomposition \cite[Ch. 9]{Kunze71}, matrix $\fSigma$ for a real $p$ can be diagonalized as follows:	
	$$\fSigma^p = \fQ \Lambda^p \fQ^T,$$
	where $\fQ$ is an orthogonal matrix, i.e., $\fQ^T \fQ = I$ and $\Lambda^p = \mathrm{diag}(\lambda_1^p, \dots, \lambda_{\bar{n}}^p).$ Now, because $\fSigma^p$ is still symmetric positive definite we have $\sigma_t(\fSigma^p) = \lambda_t(\fSigma^p) = \lambda_t^p$ for every $t \in [\![\bar{n}]\!].$ Next, raising the double bound given in \eqref{Ineq.DB_Sing} to power $p$ we have
	% ---
	\begin{align}
		\label{Ineq.DB_Sing_Mat_Pow}
		C_{\sigma_{\rm min}}^p \bar{n}^{-p\mu} \le \sigma_t(\fSigma^p) \le C_{\sigma_{\rm max}}^p \bar{n}^{p\mu/2},
	\end{align}
	% ---
	which implies $\sigma_t(\fSigma^p) \in \Omega(\bar{n}^{-p\mu}) \cap \mathcal{O}(\bar{n}^{p\mu/2}).$ Next, we extend these results for negative powers, i.e., when $p <0.$ Observe that in these cases
	$\fSigma^p = \fQ \Lambda^p \fQ^T$ where $\lambda_t^p = \lambda_t^{-|p|}.$ Then, raising to power $|p|$ and taking reciprocal the double bounds in \eqref{Ineq.DB_Sing_Mat_Pow} gives
	% ---
	\begin{align}
		%\label{Ineq.DB_Sing_Mat_Pow}
		C_{\sigma_{\rm max}}^{-|p|} \bar{n}^{-|p|\mu/2} \le \sigma_t(\fSigma^p) \le C_{\sigma_{\rm min}}^{-|p|} \bar{n}^{|p|\mu}.
	\end{align}
	% ---
	\end{proof}
	% ---

	\section*{}
	% ---
	\bibliographystyle{IEEEtran}
	\bibliography{Lit}
	
	% ---
\end{document}